\documentclass[a4paper,USenglish,cleveref,autoref,thm-restate]{lipics-v2021}
\usepackage{algorithm}
\usepackage[noend]{algpseudocode}
\usepackage{xspace}
\usepackage{multicol}
\usepackage{fvextra}

\newcommand{\Z}{\mathbb{Z}}
\newcommand{\R}{\mathbb{R}}
\newcommand{\Rp}{\mathbb{R}_{\ge 0}}
\newcommand{\N}{\mathbb{N}}

\newcommand{\veca}{\mathbf{a}}
\newcommand{\vecx}{\mathbf{x}}
\newcommand{\vecy}{\mathbf{y}}
\newcommand{\vecz}{\mathbf{z}}
\newcommand{\vecw}{\mathbf{w}}
\newcommand{\vecv}{\mathbf{v}}
\newcommand{\vecc}{\mathbf{c}}
\newcommand{\vece}{\mathbf{e}}
\newcommand{\vech}{\mathbf{h}}
\newcommand{\vecmu}{\boldsymbol{\mu}}
\newcommand{\vecp}{\mathbf{p}}
\newcommand{\vecq}{\mathbf{q}}
\newcommand{\veczero}{\mathbf{0}}

\newcommand{\supp}{\mathrm{supp}}
\newcommand{\Supp}{\mathrm{Supp}}

\newcommand{\Normaliz}{\texttt{Normaliz}\xspace}

\pdfoutput=1
\hideLIPIcs
\nolinenumbers
\title{Scalable Enumeration of Pareto-optimal Polymers for Computing Equilibrium Concentrations}
\titlerunning{Scalable Enumeration of Pareto-optimal Polymers}

\author{Archit Patil}{University of Texas at Austin, USA}{architpatil@utexas.edu}{https://orcid.org/0009-0007-8094-0762}{}
\author{Minki Hhan}{Korea Advanced Institute of Science and Technology, South Korea}{minkihhan@kaist.ac.kr}{https://orcid.org/0000-0001-5143-4587}{}
\author{David Soloveichik}{University of Texas at Austin, USA}{david.soloveichik@utexas.edu}{https://orcid.org/0000-0002-2585-4120}{}

\authorrunning{A. Patil, M. Hhan, and D. Soloveichik}

\Copyright{Archit Patil, Minki Hhan, and David Soloveichik}



\category{} 
\relatedversion{} 

\funding{Schmidt Sciences Polymath Award to D.S., DOE grant DE-SC0024248, NSF SemiSynBio III: GOALI grant 2227578}

\supplement{The supplementary materials include the source code, benchmark scripts, experimental input files, and generated data.}

\supplementdetails[subcategory={Github Repo}, cite={}, swhid={}]{Software}{https://github.com/architrahul/Pareto-polymer-enumerator}

\acknowledgements{We thank And Kaan Ata Yilmaz for valuable discussions and insights.}

\EventEditors{John Q. Open and Joan R. Access}
\EventNoEds{2}
\EventLongTitle{42nd Conference on Very Important Topics (CVIT 2016)}
\EventShortTitle{CVIT 2016}
\EventAcronym{CVIT}
\EventYear{2016}
\EventDate{December 24--27, 2016}
\EventLocation{Little Whinging, United Kingdom}
\EventLogo{}
\SeriesVolume{42}
\ArticleNo{23}
\begin{document}

\maketitle

\begin{abstract}
Predicting equilibrium concentrations of molecular complexes is essential for verifying the behavior of engineered DNA systems.
However, a finite set of monomer types can in principle generate infinitely many complexes.
We study this candidate-enumeration problem in a geometry-free, domain-level abstraction called a domain-monomer system, generalizing Thermodynamic Binding Networks (TBNs) to the unsaturated setting where not every possible bond need be formed.
We define \emph{Pareto-suboptimal} polymers as those that can be split into non-interacting parts, and show that restricting attention to \emph{Pareto-optimal} polymers is thermodynamically justified: no Pareto-suboptimal polymer appears in any minimum free-energy configuration, and the total equilibrium concentration of such polymers is small.
We prove that there are finitely many Pareto-optimal polymers and exactly characterize them via a Hilbert basis computation, extending prior work from the saturated TBN model. To scale this approach to large systems, we develop a framework that restricts the number of different monomer types that a single polymer contains, and uses combinatorial covering designs to reduce the number of Hilbert basis computations required.
We benchmark the method on several families of DNA molecular programming systems, demonstrating order-of-magnitude speedups over direct computation while recovering nearly all equilibrium-relevant polymers.
\end{abstract}

\section{Introduction}
Understanding and engineering complex molecular systems and behaviors requires both kinetic and thermodynamic modeling.
In the context of dynamic DNA nanotechnology,
NUPACK~\cite{dirks2007thermodynamic,NUPACK} is widely used to analyze the minimum free energy configurations of complexes, and their concentrations at thermodynamic equilibrium.
However, since NUPACK operates at the individual base-pair level, it is infeasible to use for large systems (many strands and/or large complex sizes).
At the same time, molecular programming researchers have been developing systems of ever larger complexity.
For example, recent work demonstrating DNA neural networks and logic circuits involved 700 distinct complexes and over 200 strands in a single test tube~\cite{cherry2025supervised,song2025heat}.
Such complexity necessitates a higher level of abstraction.

Most engineered DNA systems are organized around domains, contiguous sequence regions that behave as logical binding units. Domain-level abstractions make different geometric commitments: some models enforce rigid geometry (e.g., tile assembly), while others ignore geometry and retain only which domain types can bind. In this paper we study the latter, geometry-free setting. Our formalism is a \emph{domain-monomer system}, closely related to Thermodynamic Binding Networks (TBNs)~\cite{TBNpaper,breik2019computing,breik_kinetics,haley2021computing,wang2026molecular}, in which a monomer is a multiset of abstract domains, complementary domains can bind, and a polymer is a multiset of monomers. Unlike the saturated TBN setting which is meant to capture strong bonding, we do not require every possible bond to be formed.
The solution might have polymers with unbound complementary domains, corresponding to a more realistic model that permits bonds of varying affinities.
The central thermodynamic question is to determine the equilibrium distribution of polymer concentrations: which polymers form, and in what amounts? In engineered systems, some polymers are intended products and verifying correct behavior requires checking their equilibrium concentrations (e.g., that certain polymers are present in specific ratios), while others are undesired leakage.

In typical systems there is a relatively small set of polymers that dominate the equilibrium, with the rest present in negligible concentrations (a finite set of monomer types can generate infinitely many polymers).
If we could a priori eliminate the vast majority of safe-to-neglect polymers, then equilibrium optimization itself would be efficient (tools such as COFFEE~\cite{COFFEE} can handle hundreds of thousands of candidate polymers).
But constructing a sensible finite candidate set of polymers in a principled way appears nontrivial as it must be defined independently of the unknown equilibrium solution itself.
We address this difficulty with a natural structural notion called \emph{Pareto-optimality}. Informally, a polymer is Pareto-suboptimal if it can be split into two non-complementary polymers.
More precisely, a polymer is Pareto-suboptimal if it is possible to reconfigure the monomers of the polymer into two or more separate polymers without decreasing overall bonding.
Such a split preserves the total bond energy (enthalpy) while producing more separate complexes (increasing entropy), so Pareto-suboptimal polymers are thermodynamically disfavored.
We formally justify the focus on Pareto-optimal polymers by showing that no Pareto-suboptimal polymer appears in a minimum free-energy configuration in a discrete model, and further that in the continuous equilibrium model their total concentration is bounded in terms of the total concentration of Pareto-optimal polymers.
This makes Pareto-optimal polymers a principled candidate set for equilibrium analysis.

Prior work of Haley and Doty~\cite{haley2021computing} connected Hilbert bases to polymer enumeration in the saturated TBN model, which takes the thermodynamic limit of strong bonds. 
Our goal is to extend this algebraic perspective to the more general unsaturated setting of possibly weak bonds, allowing different bond strengths corresponding to hybridization energies of different domains.
Further, we want to make the analysis practical on large systems where direct Hilbert basis computation is far too expensive.

Our contributions are as follows.
\begin{itemize}
    \item We bound the total concentration of Pareto-suboptimal polymers in geometry-free domain-monomer systems, justifying their exclusion.
    \item We show that the full set $P^*$ of Pareto-optimal polymers is finite and can be characterized exactly via Hilbert basis enumeration. This extends the Hilbert-basis connection from the saturated TBN, discrete setting~\cite{haley2021computing}.
    \item We develop a scalable framework for enumerating the subset $P_t^* \subseteq P^*$ of Pareto-optimal polymers that use at most $t$ monomer types (\emph{support-boundedness restriction}).
    We first express $P_t^*$ using Hilbert bases of smaller subsystems, and reduce the number of required subsystem computations by using covering designs. The resulting method returns a superset $\hat{P}_{k,t}$ satisfying $P_t^* \subseteq \hat{P}_{k,t} \subseteq P^*$.
Importantly, the parameter $t$ is interpretable: rather than imposing an opaque algebraic or numerical cutoff, it bounds the number of distinct monomer types in a polymer, a quantity often apparent from the design specification of the intended products and plausible leakage pathways.
    \item We describe implementation techniques, including a probe-and-prune heuristic for choosing the block size $k$, and benchmark the method on several families of molecular systems from the recent literature~\cite{wang2026molecular,yilmaz2026modularity,sterin2025thermodynamically}. For example, on a 7-module linear cascade, the covering-design strategy with $(t,k)=(5,25)$ reduces polymer enumeration time from $>$1000\,s for direct full Hilbert basis computation to 24\,s, without missing any equilibrium-relevant polymer in the corresponding leakage analysis.
\end{itemize}

Our results provide one theoretical and algorithmic answer to the candidate-enumeration problem for equilibrium analysis in large geometry-free molecular systems: Pareto-optimal polymers give a well-justified reduced search space characterized by a Hilbert basis enumeration, and covering-design-based support restriction makes computation feasible in practice.

\section{Preliminaries}
\subsection{Model}
Let $\Z,\N,\R$, and $\Rp$ denote the sets of integers, nonnegative integers, real numbers, and nonnegative real numbers, respectively.
For a finite set $A$, we define $\N^A, \Z^A, \R^A,$ and $\Rp^A$ as the sets of $|A|$-dimensional vectors over $\N, \Z, \R,$ and $\Rp$, respectively. The vector $\vecx \in \N^A$ (or $\Z^A, \R^A, \Rp^A$) is represented by $(\vecx (a))_{a \in A}$, or sometimes $(\vecx_a)_{a \in A}$ if there is no confusion.
A multiset $\vecx$ over $A$ is an unordered collection of elements allowing duplicates, which is naturally identified by a vector $\vecx$ in $\N^A$.
Similarly, we define $\Z^{A\times B}$ as the set of $|A|\times|B|$-matrices over $\Z$ for finite sets $A, B$. The matrix-vector multiplication $M\vecx$ for $M\in \Z^{A\times B}$ and $\vecx \in \Z^B$ (or $\N^B, \R^B, \Rp^B$) is defined in a natural way.

We use the following three-level abstraction motivated by DNA nanotechnology and the TBN model~\cite{TBNpaper,breik2019computing}. This model captures how monomers (indivisible molecules) are described by domains, and how they combine to form polymers (complexes).
\begin{definition}
    A \emph{domain-monomer system} is a pair $(\Sigma,M)$ where
\begin{itemize}
    \item $\Sigma = D \cup D^*$ is a set of domains described by $D$ and $D^*$, which are the disjoint and same-sized sets of primal and complementary domains, respectively. For each $a \in D$, there exists a corresponding $a^* \in D^*$ that can form a bond with $a$, i.e., $D^* = \{a^* : a \in D\}$. We say $a$ and $a^*$ are complementary to each other, and call $(a,a^*)$ the domain pair.
    \item $M$ is a finite set of monomers, each of which is a finite multiset of domains over $\Sigma$.
\end{itemize}
\end{definition}
Each monomer $m \in M$ is encoded as a vector $\mathbf{a}_m \in \mathbb{Z}^{D}$, where the entry corresponding to the domain pair $\sigma = (a, a^*) $ is $+1$ for each occurrence of $a$ in $m$ and $-1$ for $a^*$.
This gives the \emph{monomer matrix} $A \in \mathbb{Z}^{D \times M}$, whose columns $\mathbf{a}_m$ correspond to each $m \in M$.

Given a domain-monomer system $(\Sigma,M)$, a polymer is defined as a finite, nonzero multiset of monomers.
We represent each polymer as a vector $\vecp \in \mathbb{N}^{M}\setminus \{\veczero\}$ where $\vecp(m)$ (or sometimes $\vecp_m$) denotes the number of copies of monomer $m$ in polymer $\vecp$.
In other words, $P=\N^M\setminus \{\veczero\}$ defines the set of all possible polymers.
The net domain vector of $\vecp$ is $A\vecp \in \mathbb{Z}^{D}$, whose entry for domain $a$ gives the
net count of $a$ minus $a^*$ across all monomers in the polymer. We call an equation $\vecp=\vecp_1+\vecp_2$ a split (of $\vecp$).

Intrapolymer bonding is assumed throughout: complementary domains within a single polymer are taken to be already bound. Under this assumption, the free domains of a polymer $\vecp$ are precisely those whose corresponding entry in $A\vecp$ is nonzero.
Note that as in the TBN model, there is no underlying notion of geometry which can prevent two complementary domains within a polymer from binding.

\subsection{Computing Polymer Concentrations at Equilibrium}
The \emph{concentration} $[\vecp]$ of polymer $\vecp \in P$ specifies the amount of polymer $\vecp$ in the system.
In our theoretical results, we use dimensionless mole fractions~\cite{dirks2007thermodynamic},
capturing the amount of the polymer relative to the solvent molecules (water).
In normal dilute solutions, $\sum_{\vecp \in P} [\vecp] \ll 1$.
For each $m\in M$, the polymer concentration induces the total monomer concentration
\begin{equation}\label{eqn:monomer_concentration}
    [m]=\sum_{\vecp \in P} \vecp_m [\vecp].
\end{equation}

Equilibrium analysis aims to determine the concentrations of polymers at equilibrium, given an initial supply of monomers (or polymers).
NUPACK's {\tt complex\_concentrations} function~\cite{NUPACK}\footnote{Equivalently the {\tt concentrations} command line tool for older NUPACK 3.} and COFFEE~\cite{COFFEE} can compute equilibrium concentrations given the initial monomer concentrations and a candidate set of polymers. However, the main computational bottleneck is typically not the equilibrium computation itself, but the construction of an appropriate candidate set of polymers.
For a given domain-monomer system, the set of possible polymers is generally infinite and cannot be enumerated exhaustively.

In this paper, we introduce the principled restriction to polymers that are Pareto-optimal.
As we show next, Pareto-suboptimal polymers can be excluded on thermodynamic grounds. This reduces the search space from an intractable infinite family to a finite set of
candidates, after which additional enumeration strategies can be applied.

For complex systems, even restricting to Pareto-optimal polymers is computationally intractable.
This motivates our subsequent methodology of limiting the ``size'' of the enumerated Pareto-optimal polymers, where size can be defined in two ways:
(1) the number of different types of monomers, or
(2) the number of types of domain pairs in the polymer.

\begin{remark}
A domain-monomer system corresponds exactly to a Thermodynamic Binding Network~(TBN)~\cite{TBNpaper,breik2019computing} under two relaxed assumptions. First, saturation is not assumed: the net domain vector $A\mathbf{x}$ of a polymer $\mathbf{x}$ need not be nonnegative, meaning unstarred domains may remain unbound. Second, no assumption is made that enthalpy dominates entropy. The present work analyzes TBNs from this more general perspective, without either idealization.
\end{remark}

\section{Pareto-Optimal Polymers}
\subsection{Defining Pareto-Optimal Polymers}

Throughout this paper, we fix the domain-monomer system $(\Sigma,M)$ unless specified otherwise. 
In this section, we introduce reasoning that allows us to exclude certain polymers from consideration, based on the observation that they do not occur in minimum free energy configurations. This leads to a \emph{finite} search space instead of the infinite set of all polymers.

\begin{definition}[Complementary polymers]
Two polymers $\vecp,\vecq\in P$ are complementary if there exists a domain such that the entries of $A\vecp$ and $A\vecq$ corresponding to that domain are both nonzero and have opposite sign. In other words, $(A\vecp)_a \cdot (A\vecq)_a <0$ for some $a\in D$.
\end{definition}

Equivalently, $\vecp$ and $\vecq$ are complementary if and only if the combined polymer $\vecp+\vecq$ contains at least one intrapolymer bond between a domain of $\vecp$ and a domain of $\vecq$, so that splitting $\vecp+\vecq$ back into $\vecp$ and $\vecq$ necessarily breaks at least one bond.

\begin{definition}[Pareto-(sub)optimal polymer]
A polymer $\vecp \in P$ is \emph{Pareto-suboptimal} if there exists a split $\vecp = \vecp_1 + \vecp_2$ with $\vecp_1, \vecp_2 \in P$ such that $\vecp_1$ and $\vecp_2$ are not complementary. If $\vecp$ is not Pareto-suboptimal, it is \emph{Pareto-optimal}. The set of all Pareto-optimal polymers is denoted $P^*$.
\end{definition}

\Cref{fig:suboptimal-example} illustrates a Pareto-suboptimal polymer: its monomers can be regrouped into two non-complementary polymers without breaking any bond.
We stress that the model does not track bonds between specific domains.
A polymer is represented only by its monomer multiset $\vecp$; the explicit bonds drawn in the figure are an informal visualization rather than part of the model.
From the model's perspective, the reconfiguration changes only how the monomers are partitioned into separate polymers, leaving the bonding unchanged.
Note that the geometry-free nature of the model is essential here.
Geometric constraints could permit $\vecp$ but not its decomposition into $\vecp_1 + \vecp_2$.

\subsection{Thermodynamic Justification}

In this section we justify the exclusion of Pareto-suboptimal polymers on thermodynamic grounds, arguing that they are in a certain sense negligible.
We begin with the discrete model, showing that no Pareto-suboptimal polymer appears in any minimum free energy configuration.
This provides strong motivation for the exclusion, but the gap between discrete minimum free energy configurations and continuous equilibrium concentrations calls for a quantitative bound in the continuous regime.
Accordingly, in the second part of this section we derive a general bound on the total concentration of suboptimal polymers in terms of the total concentration of optimal ones.

\subsubsection{Discrete thermodynamics}
Towards justifying the exclusion of Pareto-suboptimal polymers, in this section we use a simple discrete model to show that any minimum free energy configuration contains only Pareto-optimal ones. Given an initial composition, i.e., a multiset of monomers, a \emph{configuration} is a multiset of polymers built from this composition. The \emph{free energy} of a configuration $C$ is \[G(C) \;=\; a \cdot H(C) \;-\; b \cdot S(C)\]
where $H(C)<0$ is the total enthalpy (sum of bond energies across all bonds in $C$), $S(C)$ is the entropy (number of separate polymer units), and $a, b > 0$ are system-dependent weights. The probability of observing a configuration $C$ is proportional to its Boltzmann weight:
  \[\mathbb{P}[C] \;\propto\; e^{-G(C)/kT}\]

\begin{lemma}
\label{lem:paretosub-equilibrium}
No Pareto-suboptimal polymer appears in any minimum free energy configuration.
\end{lemma}
The proof is deferred to \cref{app:proofTheormo}.

\Cref{fig:pareto-curve-subfigure} shows an example configuration $C_1$ where the first polymer is Pareto-suboptimal.
Splitting this polymer yields
configuration $C_2$, which strictly improves on the free energy.
Note that a configuration may have no Pareto-suboptimal polymer and still not lie on the Pareto-optimal frontier.
For example, the second and third polymers in $C_2$ can be reconfigured into three polymers, which again strictly improves on the free energy (configuration $C_3$).

\begin{figure}[ht]
    \centering

  \begin{subfigure}[t]{0.15\textwidth}
    \includegraphics[width=\textwidth]{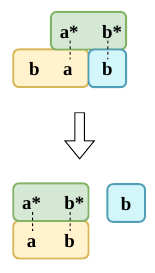}
    \caption{}
    \label{fig:suboptimal-example}
  \end{subfigure}
    \hspace{20pt}%
  \begin{subfigure}[t]{0.6\textwidth}
    \includegraphics[width=\textwidth]{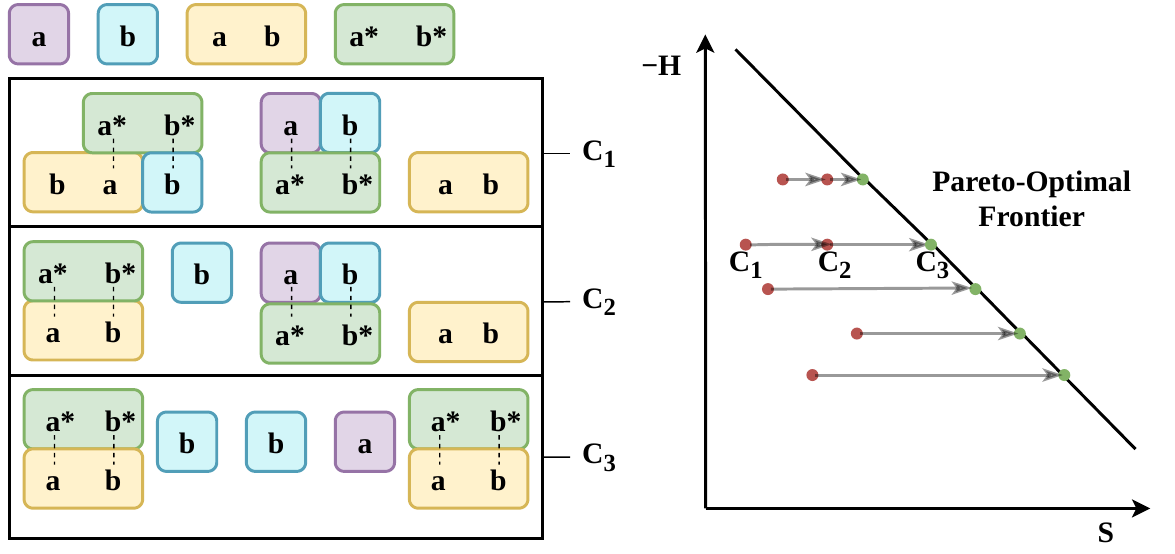}
    \caption{}
    \label{fig:pareto-curve-subfigure}
  \end{subfigure}

    \caption{Pareto-suboptimal polymers and the Pareto-optimal frontier.
    \textbf{(\subref{fig:suboptimal-example})}~Visualization of a Pareto-suboptimal polymer (top) whose monomers can be regrouped into two non-complementary polymers with the same number of bonds (bottom).
    \textbf{(\subref{fig:pareto-curve-subfigure})}~
    The Pareto-optimal frontier in the $(S, - H)$ plane.
    Minimum free energy configurations lie on the frontier.
    Configurations containing Pareto-suboptimal polymers lie strictly below the frontier;
    each such configuration is dominated by a configuration with the same
    enthalpy and strictly higher entropy (rightward arrows).
    While suboptimal polymers necessarily make a configuration suboptimal, the implication does not go the other way.
    For example, $C_2$ does not contain a Pareto-suboptimal polymer, but reconfiguring its monomers yields the configuration $C_3$ with one more polymer but the same bonds.
    }
    \label{fig:pareto-curve}
\end{figure}

\subsubsection{Continuous thermodynamics}
Toward further justifying the exclusion of Pareto-suboptimal polymers, we now turn to the continuous concentrations regime.
We show that the total concentration of Pareto-suboptimal polymers is much smaller than the total concentration of Pareto-optimal polymers at equilibrium.

The equilibrium in the continuous model corresponds to the minimizer of the pseudo-Helmholtz free energy
\[
g = \sum_{\mathbf{p} \in P} [\mathbf{p}]\bigl(\log [\mathbf{p}] + G_\mathbf{p} - 1\bigr),
\]
where
$G_\mathbf{p} = -\log \Omega_\mathbf{p}$ is the formation energy of the polymer $\vecp\in P$, subject to the monomer conservation constraints \cref{eqn:monomer_concentration}.

At equilibrium, for any (nontrivial) split $\mathbf{p} = \mathbf{p}_1 + \mathbf{p}_2$ where $\mathbf{p}_1$ and $\mathbf{p}_2$ are not complementary, the following holds for the dimensionless mole fractions in \cref{eqn:monomer_concentration}:
\begin{equation}\label{eq:factorization}
[\mathbf{p}] = [\mathbf{p}_1]\cdot [\mathbf{p}_2].
\end{equation}
This equality follows from the first-order optimality conditions of the pseudo-Helmholtz free energy minimization, which yield $[\mathbf{p}] = \Omega_\mathbf{p} \prod_{m \in M} z_m^{\mathbf{p}(m)}$ for some $z_m > 0$. The factorization over a non-complementary split then holds because the formation energy is additive across non-interacting parts. This is a standard result in the polymer equilibrium literature~\cite{dirks2007thermodynamic}.

\begin{theorem}\label{thm:suboptimal-bound}
Let $[\cdot]$ denote equilibrium concentrations, and define
$C_{\mathrm{opt}} := \sum_{\mathbf{p} \in P^*} [\mathbf{p}],$ and $
C_{\mathrm{sub}} := \sum_{\mathbf{p} \in P \setminus P^*} [\mathbf{p}].$
If $C_{\mathrm{opt}} < 1$, then
\[
C_{\mathrm{sub}} \leq \frac{C_{\mathrm{opt}}^2}{1 - C_{\mathrm{opt}}}.
\]
\end{theorem}
In particular, in the dilute regime when $C_{\mathrm{opt}} \ll 1$, we have $C_{\mathrm{sub}} = O(C_{\mathrm{opt}}^2)$, so the total concentration of Pareto-suboptimal polymers is negligible compared with the total concentration of Pareto-optimal polymers.
We defer the proof to \cref{app:proofbound}.

\section{Enumeration of Pareto-Optimal Polymers via Hilbert Basis Computation}
\label{sec:hilbert}
This section shows that the Pareto-optimal polymers of a domain-monomer system can be identified by a Hilbert basis of a set of neutral polymers having zero net domain vectors in an extended system.

Our Hilbert-basis characterization generalizes the saturated-TBN framework of Haley and Doty~\cite{haley2021computing}.
Specifically, while the original TBN model~\cite{TBNpaper,breik2019computing} used by the prior work focuses on configurations with maximum binding (called saturated configurations), our domain-monomer system relaxes this assumption and allows unpaired complementary domains in different polymers to coexist in a configuration.
In order to make this generalization, we utilize an additional augmented neutralization construction described below.

Given a domain-monomer system $(\Sigma, M)$ with monomer matrix
$A \in \mathbb{Z}^{|D| \times |M|}$, we define an \emph{augmented
(domain-monomer) system} by introducing two \emph{unit monomers} per domain pair. More precisely, the augmented system $(\Sigma,M')$ is a domain-monomer system defined as follows: It inherits the original domain $\Sigma$, and $M'$ is defined as a superset of $M$ by
\[
M' = M \cup \left\{u_a=\{a\},u_{a^*}=\{a^*\}: (a,a^*)\text{ is a domain pair}
\right\}.
\]
In words, for each domain pair $(a, a^*)$, the unit monomer $u_a$ contains only
$a$, and the unit monomer $u_{a^*}$ contains only $a^*$. Both are added as monomers in the augmented system.
Their columns in the
augmented monomer matrix are $+1$ and $-1$, respectively, in the coordinates
corresponding to $(a, a^*)$ and $0$ elsewhere. Let $M'$ denote the augmented
monomer set and $A' \in \mathbb{Z}^{D \times M'}$ the corresponding
augmented monomer matrix.
It is worth noting that for $\vecw=(\vecw_M,\vecw_\Sigma) \in \Z^{M}\times \Z^{\Sigma} = \Z^{M'}$, it holds that for $a\in D$
\begin{equation}\label{eqn:A'expanding}
(A'\vecw)_a=(A\vecw_M)_a +(\vecw_{\Sigma})_{u_a}-(\vecw_{\Sigma})_{u_{a^*}}.
\end{equation}

For simplicity, we assume that the original monomer set $M$ does not contain any unit monomer $u_a$ or $u_{a^*}$ so that $|M'|=|M|+|\Sigma|$, which significantly simplifies the details of this section.
However, our results continue to hold without this assumption with minor modification of the argument.\footnote{For example, we can add ``dummy'' domains to the singleton monomers in the original system to enforce the assumption.}

A polymer $\vecx \in \Z^{M}$ can be \emph{neutralized} by attaching the unit monomers to unbound domains. More precisely, we define the neutralizing map $\nu:\Z^{M}\to \Z^{M'}$ as follows.
\[
(\nu(\vecx))_m=
\begin{cases}
    \vecx_m \text{ for }m\in M\\
    \max\left(-(A\vecx)_a,0\right) \text{ for }m=u_a\\
    \max\left((A\vecx)_a,0\right) \text{ for }m=u_{a^*}
\end{cases}
\]
In words, for each pair $(a,a^*)$, we add $c$ copies of $u_{a^*}$ if $\vecx$ has $c>0$ unbound domain $a$, and add $c$ copies of $u_{a}$ if $\vecx$ has $c>0$ unbound $a^*$. \cref{fig:neutralizing} visualizes the neutralizing map.

\begin{figure}[ht]
    \centering
    \includegraphics[width=0.7\linewidth]{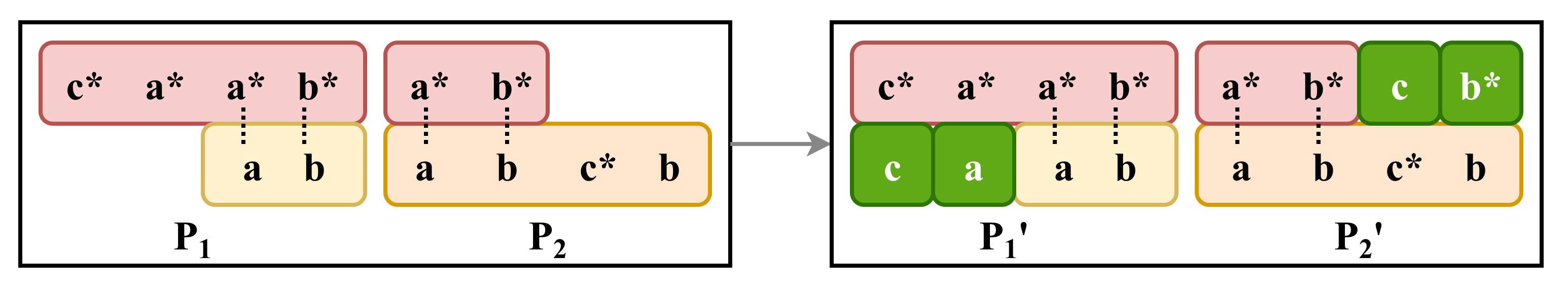}
    \caption{The polymers in the first configuration have been neutralized by unit monomers so that each vector representation has zero net domain vector.}
    \label{fig:neutralizing}
\end{figure}
We give the algebraic interpretation.
For each $a\in D$, let $e_a\in \mathbb{Z}^D$ denote the unit vector with a $1$ in the coordinate corresponding to $a$ and $0$ in all other coordinates. Observe that the neutralizing map indeed makes the net domain vector zero:
\[
A'\cdot \nu(\vecx)
=
A\vecx
+
\left(\sum_{a \in D}\max\left(-(A\vecx)_a,0\right) \cdot e_a\right)
+
\left(\sum_{a \in D}\max\left((A\vecx)_a,0\right) \cdot (-e_a)\right)
=
\veczero.
\]
This follows from \cref{eqn:A'expanding} coordinate-wise for each $a\in D$.
This means that all the domains are bound, justifying the name of the neutralizing map $\nu$.
In general, we define the set of \emph{neutral} polymers\footnote{In other words, the saturated polymers. We stress that we do not assume that all polymers are saturated; we consider this set only in the augmented system for deriving the computational tools.} in the augmented system:
\[
  \mathcal{C}(A') \;=\; \bigl\{\, \vecz \in \mathbb{N}^{M'}
  \;\bigm|\; A'\vecz = \mathbf{0} \,\bigr\}.
\]
Each $\vecz \in \mathcal{C}(A')$ has a zero net domain vector. The set of neutral polymers has a clean algebraic representation, allowing us to relate the Pareto-optimal polymers in the system with the ``indecomposable'' elements, i.e., the Hilbert basis elements as we will see later.

Given a neutral polymer in the augmented system,
we can recover a polymer in the original system using a projection map $\pi:\Z^{M'}=\Z^M\times \Z^{\Sigma}\to \Z^{M}$ that discards the unit monomer coordinates. In particular, $\pi(\nu(\vecx))=\vecx$ holds for all polymers $\vecx\in\Z^M$.

It is clear that $(\nu(\vecx))_{u_a}\cdot (\nu(\vecx))_{u_{a^*}}=0$ for all $a\in D$. We call a polymer $\vecw \in \Z^{M'}$ \emph{dichotomous} if it satisfies this property: $(\vecw)_{u_a}\cdot (\vecw)_{u_{a^*}}=0 \text{ for all }a \in D$.
The following lemma is not hard to show.
Intuitively, among all neutral augmented polymers that project to a fixed original polymer $\vecx$,
there is exactly one with no redundant canceling pair of unit monomers.

\begin{lemma}\label{lem:unique_lifting}
    $\nu(\vecx)$ is the unique dichotomous neutral polymer in $\pi^{-1}(\vecx) \cap \mathcal C(A')$.
\end{lemma}

\subsection{Hilbert Basis}
We briefly recall the definitions and basic properties of Hilbert bases that we need. First, for an integer matrix $B\in \Z^{m\times n}$, the set $C=\{\vecx \in \mathbb R^n: B\vecx \ge \veczero\}$ defined by a system of inequalities is called a (rational polyhedral) cone. If there exists a $\vecc\in \mathbb R^n$ such that $\vecc \cdot \vecx>0$ for all $\vecx \in C\setminus \{\veczero\}$, we say $C$ is a pointed cone.

The following definition of the Hilbert basis will be convenient for this paper.
\begin{definition}[Hilbert basis]
The \emph{Hilbert basis} $H(C)$ of a pointed cone $C$
is a subset of $C \cap \mathbb Z^{n} \setminus \{\veczero\} $ defined as follows:
\[
H(C)= \bigl\{
\vech \in C \cap \mathbb Z^{n} \setminus \{\veczero\} \bigm|\vech\text{ is not the sum of two vectors in }C \cap \mathbb Z^{n} \setminus \{\veczero\}
\bigr\}.
\]
\end{definition}
It is known that the Hilbert basis of any pointed (rational polyhedral) cone is 
finite.
There are several equivalent definitions of the Hilbert basis.
Notably, it is the smallest set that generates $C \cap \mathbb Z^{n}$ by nonnegative integer combinations;
this characterization was previously used to enumerate the minimal set of reactions or polymers \cite{akef2025computing,haley2021computing}.\footnote{Technically, our definition coincides with this usual definition only for the \emph{pointed} cone.}
The computation of the Hilbert basis has been extensively studied in the literature \cite{bruns2010normaliz,hemmecke2002computation,pottier1996euclidean}, and implementations are publicly available \cite{4ti2,Normaliz}.

\subsection{Pareto-Optimal Polymers as Hilbert Basis}
The set of neutral polymers $\mathcal{C}(A')$ is an intersection of $\mathbb Z^{M'}$ and a cone
\begin{equation}\label{eqn:our_pointedcone}
\mathcal{C}_{\mathbb R}(A')  = \bigl\{\vecx \in  \mathbb R^{M'}\bigm|A'\vecx \ge\veczero ,-A'\vecx\ge\veczero, I\vecx=\vecx\ge\veczero\bigr\}
\end{equation}
which is pointed because $\mathbf 1\cdot \vecx>0$ holds for all nonzero $\vecx \in \mathcal{C}_{\mathbb R}(A') $.
We denote $H=H(\mathcal{C}_{\mathbb R}(A'))$ the Hilbert basis of $\mathcal{C}_{\mathbb R}(A')$, which is a subset of $\mathcal{C}(A')=\mathcal{C}_{\mathbb R}(A') \cap \Z^{M'}$ by definition.

Recall $P^*$ is the set of Pareto-optimal polymers in the original system $(\Sigma,M)$.
The following theorem states that the Pareto-optimal polymers are characterized by the projection of Hilbert basis elements $\pi(H):=\{\pi(\vech) \;|\; \vech \in H\}$.
\begin{theorem}\label{thm:HilbertCharacterize}
    $P^*=\pi(H)\setminus \{\veczero\}.$
\end{theorem}
The proof is deferred to \cref{app:proofofHilbertCharacterize}.
A few implications are in order. First, this shows that $P^*$ is finite because of the finiteness of the Hilbert basis. Furthermore, we can, in principle, compute the set $P^*$ using the well-established Hilbert basis algorithms.
We discuss the limitations of using this theorem directly in \cref{app:failed}.

\section{Scalable Enumeration with Bounded Support}
\label{sec:covering}
We present scalable algorithms to compute a meaningful subset of Pareto-optimal polymers.
To scale the enumeration of Pareto-optimal polymers, we consider the subset of polymers with a bounded number of monomer types.
We define the \emph{support} of polymer $\vecp \in \N^M$ by
\[
  \supp(\vecp) \;=\; \{\, m \in M \mid \vecp_m > 0 \,\},
\]
and for a user-chosen parameter $t$, the \emph{$t$-bounded Pareto-optimal set} by
\[
  P^*_t \;=\; \{\, \mathbf{x} \in P^* \mid |\mathrm{supp}(\mathbf{x})| \leq t
  \,\}.
\]

We also consider a similarly restricted subset of $P^*$ based on the number of domain pairs instead of monomers. This introduces the domain-support mode, which proved ineffective in our numerical experiments. One reason is that many polymers have a larger number of domain pairs than distinct monomers. In other words, the domain mode may be advantageous when the number of domain pairs is smaller than the number of monomers, or when many monomers share the same domain pairs.
We outline this approach in \cref{sec:domainmode}.

In \cref{subsec:enum_all}, we show how to compute the $t$-bounded Pareto-optimal set through the computations of smaller Hilbert bases. However, the number of Hilbert bases to compute is large, making the naive algorithm less appealing.
Instead, we introduce our main algorithm in \cref{subsec:enum_covering}, which optimizes the overall time complexity by reducing
the number of Hilbert basis computations using covering designs.
In numerical experiments of \cref{sec:benchmarking}, we will see that the equilibrium concentration analysis with the subset $P^*_t$ is almost as good as the one with the full set $P^*$ for many applications.

\subsection{Enumeration with All \texorpdfstring{$t$}{t}-sized Supports}\label{subsec:enum_all}

We introduce a few pieces of mathematical notation that will streamline the arguments to follow.
Let $S\subseteq M$ be a subset of monomers with $|S| = t$.
The set $(\Sigma,S)$ gives a natural domain-monomer (sub)system.
Define an \emph{embedding} $\iota_S:\N^S \to \N^M$ that appends $0$ on each entry corresponding to $M\setminus S$.
After embedding, the set of polymers becomes
\[\Supp(S)=\{\vecp \in \N^{M}:\supp(\vecp)\subseteq S\}.\]
It is not hard to see that if $\vecp \in \mathrm{Supp}(S)$ admits a split $\vecp = \vecp_1 + \vecp_2$ with nonzero $\vecp_1, \vecp_2 \in \mathbb{N}^M$, then $\vecp_1, \vecp_2 \in \mathrm{Supp}(S)$ as well, since the non-negativity of entries in $\vecp$ forces $\mathrm{supp}(\vecp_1), \mathrm{supp}(\vecp_2) \subseteq \mathrm{supp}(\vecp) \subseteq S$.

Define $P_S^* = P^* \cap \Supp(S)$ as the set of Pareto-optimal polymers whose support is contained in $S$.
The following shows the connection between the Pareto-optimal polymers in the full system $(\Sigma,M)$ and subsystem $(\Sigma,S)$.
Intuitively, Pareto-optimality of the polymers supported only on $S$ can be checked entirely inside the subsystem $(\Sigma,S)$, as any split of such a polymer is automatically
also supported on $S$.
We defer the formal proof to \cref{app:proofcovring}.
\begin{lemma}\label{lem:Paretocharacter_embedding}
    $P_S^*$ is the set of Pareto-optimal polymers of the system $(\Sigma,S)$ after embedding.
\end{lemma}

In particular, this lemma implies the following inclusion:
\begin{equation}\label{eqn:inclusion}
    P_S^* \subseteq P_T^* \text{ if } S \subseteq T.
\end{equation}
Given this, we compute a subset of Pareto-optimal polymers, namely $P^*_S$, by computing a much smaller Hilbert basis.
Formally, we define the set of neutral polymers restricted to $S$ by
\[
  c^S(A') \;=\; \bigl\{\, \mathbf{x} \in \mathbb{N}^{M'}
  \;\bigm|\; A'\mathbf{x} = \mathbf{0},~~\supp(\pi(\vecx))\subseteq S \,\bigr\},
\]
and the corresponding pointed cone by
\begin{equation}\label{eqn:pointed_smaller}
    c^S_{\mathbb R}(A') \;=\; \bigl\{\, \mathbf{x} \in \mathbb{R}^{M'}
  \;\bigm|\; A'\mathbf{x} = \mathbf{0},\vecx \ge \veczero ,\vecx_m=0 \text{ for }m\in M\setminus S\,\bigr\}
\end{equation}
following \cref{eqn:our_pointedcone}.
We denote $H_S=H(c^S_{\mathbb R}(A'))$ the Hilbert basis of $c^S_{\mathbb R}(A')$. Applying \cref{thm:HilbertCharacterize} on the system $(\Sigma,S)$, we have the following Hilbert basis characterization.
\begin{lemma}\label{lem:union_naive}
    $P_S^* = \pi(H_S)\setminus \{\veczero\}$. In particular, $P_t^* = \bigcup_{\substack{S \subseteq M \\ |S| = t}} \left(\pi(H_S) \setminus\{\veczero\}\right).$
\end{lemma}
The proof is deferred to \cref{app:proofcovring}.
This lemma says that each support-restricted Pareto-optimal set can be computed by applying the Hilbert-basis characterization to the subsystem on $S$, and the $t$-bounded set is obtained by taking the union over these.

For $|S|\ll |M|$, computing a single Hilbert basis $H_S$ can be done much more efficiently than computing the full Hilbert basis $H$.
The effective ambient space of the pointed cone in \cref{eqn:pointed_smaller} is much smaller than that of the original cone in \cref{eqn:our_pointedcone}, since we can ignore the always-zero original-monomer coordinates for $m\in M\setminus S$. More precisely, the dimensions are $|S|+|\Sigma|$ and $|M|+|\Sigma|$, respectively. Some domains may not appear in the monomers in $S$, so the dimension can be further reduced.
Let $D_S$ be the set of domains that appear in the monomers in $S$. Let $M_S'=S \cup \{u_a,u_{a^*}:a\in D_S\}$.
We also reduce the number of equations defining the cone by considering $A'_S\in \Z^{D_S \times M_S'}$, obtained by removing irrelevant rows and columns. 
This significantly reduces the number of input constraints in computing $H_S$.

By \cref{lem:union_naive}, we can compute the $t$-bounded Pareto-optimal set $P_t^*$ by multiple smaller Hilbert bases $H_S$ for $|S|=t$.
This naive subset enumeration algorithm requires one invocation of Hilbert basis computation per $t$-subset of $M$, giving a total of $\binom{n}{t}$ invocations. The factor $\binom{n}{t}$ grows rapidly with $n$ and $t$, making naive enumeration infeasible for moderate $n$ and $t$.
In the next subsection, we explore a better strategy for selecting which subsets to run Hilbert basis computation on in order to reduce the total time.

\subsection{Enumeration with Covering Design}\label{subsec:enum_covering}
The starting point of this section is \cref{eqn:inclusion}.
This inclusion implies that if we compute a Hilbert basis $H_T$ for some $T$, then we learn \emph{all} $H_S$ for all $S \subset T$. Instead of computing all $H_S$ for $|S|=t$, we compute a small number of $H_T$ with $|T|=k$ for some $k>t$.
To minimize the complexity, we want to minimize the number of different $T$'s.

The \emph{minimum covering design problem} addresses this problem. This problem asks for a collection $\mathcal{B}$ of $k$-element subsets (blocks) of $M$ of minimum cardinality such that every $t$-element subset of $M$ is contained in at least one block. More formally:
\begin{definition}[Covering design]
Let $M$ be a finite set of size $n$.
A $(n,k,t)$-\emph{covering design} $\mathcal D$ is a collection of $k$-element subsets of $M$ such that for every $t$-element subset $T \subseteq M$, there exists a
block $B \in \mathcal D$ with $T \subseteq B$. The \emph{covering number} $c(n,k,t)$ denotes the minimum cardinality of any such collection.
\end{definition}

Let $\mathcal D$ be a $(|M|,k,t)$-covering design. Given $\mathcal D$, the \emph{covering design strategy} runs
the Hilbert basis algorithm on each block $B \in \mathcal D$ and computes
\[
  \hat{P}_{k,t} \;=\; \bigcup_{B \in \mathcal D} \left( \pi(H_B)\setminus \{\veczero\}\right).
\]
The following theorem shows that the covering design strategy computes (a superset of) the $t$-bounded Pareto-optimal set. The proof can be found in \cref{app:proofcovring}.
\begin{theorem}
\label{thm:covering}
$P^*_t \subseteq \hat{P}_{k,t} \subseteq P^*$.
\end{theorem}

In the covering-design strategy, using a covering design $\mathcal D$ requires $|\mathcal D|$ Hilbert basis computations. The Sch\"{o}nheim lower bound~\cite{schonheim1964coverings} gives
\[
  c(n,k,t) \;\geq\;
  \left\lceil \frac{n}{k}
  \left\lceil \frac{n-1}{k-1} \cdots
  \left\lceil \frac{n-t+1}{k-t+1} \right\rceil \cdots
  \right\rceil \right\rceil .
\]
In practice, covering designs are vastly smaller than the family of all $k$-subsets, whose cardinality is $\binom{n}{k}$. This reduction makes the covering-design strategy substantially more attractive.
However, computing the minimal covering design or computing $c(n,k,t)$ is known to be hard \cite{crescenzi2004optimal}.
We use precomputed near-optimal covering designs, or greedy constructions if these are unavailable. Further details are described in the next section.

\section{Implementation}
\label{sec:implementation}
All code is available at \href{https://github.com/architrahul/Pareto-polymer-enumerator}{github.com/architrahul/Pareto-polymer-enumerator}. We leverage the modern Hilbert basis algorithm \Normaliz \cite{Normaliz}.

\begin{algorithm}
\caption{Covering Design Strategy}
\label{alg:covering}
\begin{algorithmic}[1]
\Require Domain-monomer system $(\Sigma, M)$ with $|M| = n$, parameters $t$ and $k$ with $t \leq k \leq n$
\Ensure $\hat{P}_{k,t}$
\State $\hat{P} \leftarrow \emptyset$
\State Fetch or construct $(n,k,t)$-covering design $\mathcal{D}_{n,k,t}$
\For{each block $B \in \mathcal{D}_{n,k,t}$}
    \State Construct subsystem $(\Sigma, B)$
    \State Compute Hilbert basis $H_B$ using \Normaliz
    \State $\hat{P} \leftarrow \hat{P} \cup \left( \pi(H_B) \setminus \{0\} \right)$
\EndFor
\State \Return $\hat{P}$
\end{algorithmic}
\end{algorithm}

For a fixed size of the monomer set $n$ and Pareto-optimal polymer support-bound $t$, the total runtime of the covering design strategy for a given $k$ is determined by the number of \Normaliz invocations $|\mathcal D_{n,k,t}|$
and the per-invocation cost, which depends on the specific sub-system. 

The details of the construction of covering designs and the block-size selection heuristic are described next.

\subsection{Covering Design Construction}
We obtain the covering design given input parameters $(n,k,t)$ as follows: First, when the parameters are sufficiently small ($n < 100$, $k \leq 25$, $t \leq 8$), we fetch the near-optimal covering designs from the precomputed covering design database, the La Jolla Covering Repository~\cite{gordon_2026_19735294}.

When the parameters are outside this range, we construct the covering design by combining smaller covering designs using dynamic programming \cite{gordon1995new}.
This approach produces reasonably small covering designs, although we cannot ensure their near-optimality. This step can be done as preprocessing independent of the target system.

For the applications in this paper, the covering designs from the La Jolla repository with $k\le 25$ are comparably efficient to those with $k>25$ constructed by dynamic programming, as observed in \cref{fig:probe-cascade}.

\subsection{Selecting the Optimal Block Size \texorpdfstring{$k$}{k}}
\label{sec:probe-and-prune}

\noindent
\cref{alg:covering} requires the user to specify two parameters: the support bound $t$ and the block size $k$. The user may choose $t$ depending on the system to be analyzed and its known properties, and increasing the parameter $t$ yields a reduced candidate set~$\hat{P}_{k,t}$ or $P^*_t$ with higher expressive power, at the cost of efficiency.

On the other hand, the choice of $k$ is less obvious. For a fixed $t$, increasing $k$ reduces the number of blocks in the $(|M|,k,t)$-covering design but makes each Hilbert basis computation more expensive. The total runtime is therefore a non-monotone function of $k$.

We address this with a \emph{probing} strategy. For each $k$ and the $(|M|,k,t)$-design $\mathcal D_k$, we can estimate the expected running time of a single Hilbert basis computation for block size $k$ by running a 100-iteration probe.
\[
\hat T(k) = |\mathcal D_k|\times (\text{estimated run time of a single Hilbert basis computation}),
\]
Starting from $k=t$, we estimate the total runtime for increasing values of $k$. We continue increasing $k$ until the estimated runtime has increased for three consecutive values of $k$, and then select the value of $k$ with the smallest estimated runtime observed.

This stopping rule is motivated by the tradeoff between the size of the covering design and the cost of each Hilbert basis computation. Increasing $k$ reduces the number of blocks in the covering design, but it also makes each \Normaliz call more expensive. So, the total runtime is expected to have an intermediate minimum. As shown in \cref{fig:probe-damien,fig:probe-cascade}, once the estimated runtime passes this minimum, it increases monotonically in our tested systems. The probing step introduces only a negligible additive runtime overhead in our applications.

 \cref{fig:probe-damien,fig:probe-cascade} illustrate the runtime as a function of $k$ for the Scaffolded DNA and AND gate linear cascade systems (\cref{sec:tbn-families}), respectively. As observed, the estimated runtime is very close to the actual runtime, and there is an optimal value of $k$ that minimizes the runtime.

\begin{figure}[ht]
  \centering
  \begin{subfigure}[t]{0.48\textwidth}
    \includegraphics[width=\textwidth]{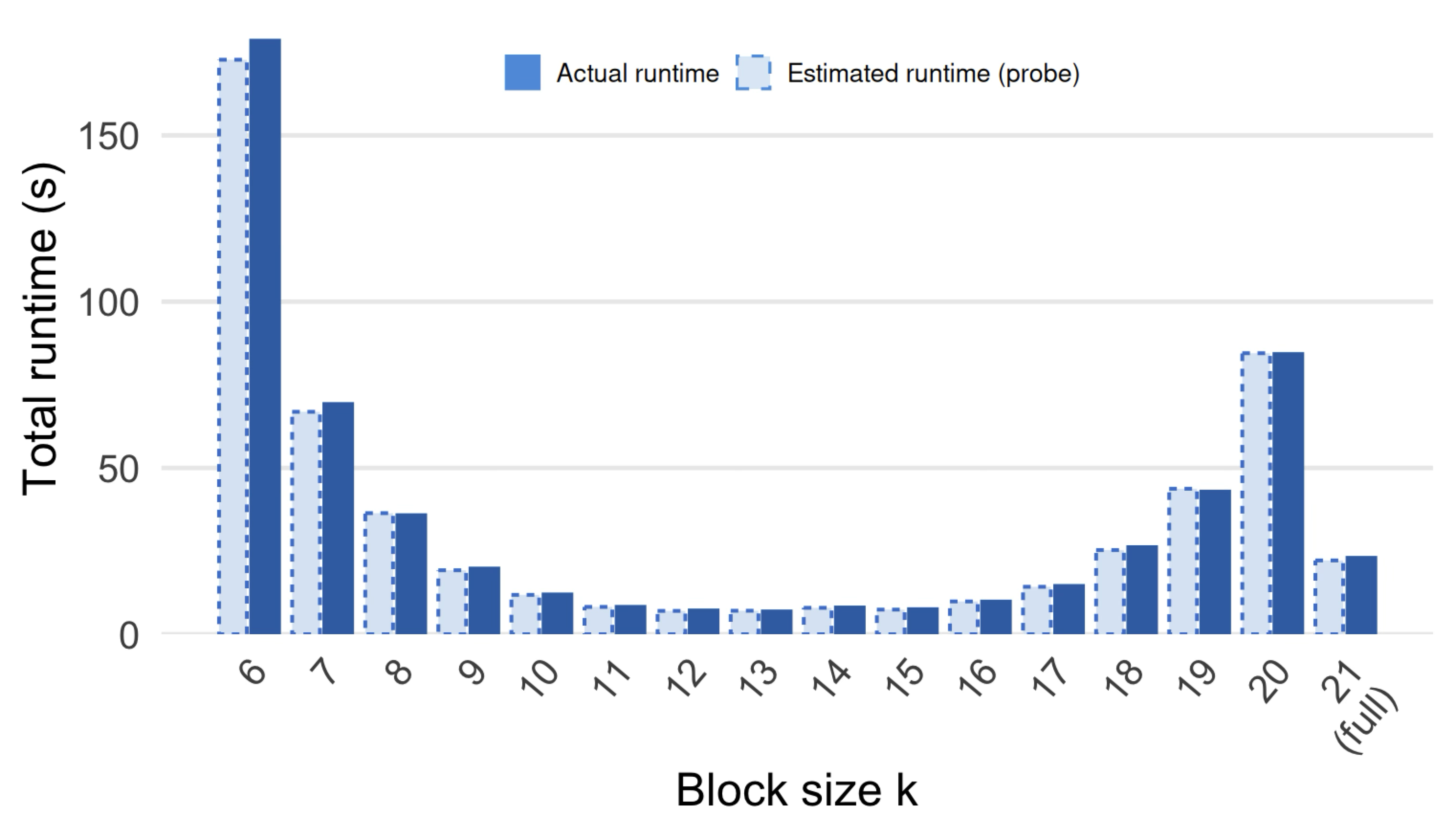}
    \caption{Scaffolded DNA system ($|M|=21$, $t=5$).}
    \label{fig:probe-damien}
  \end{subfigure}
  \begin{subfigure}[t]{0.50\textwidth}
    \includegraphics[width=\textwidth]{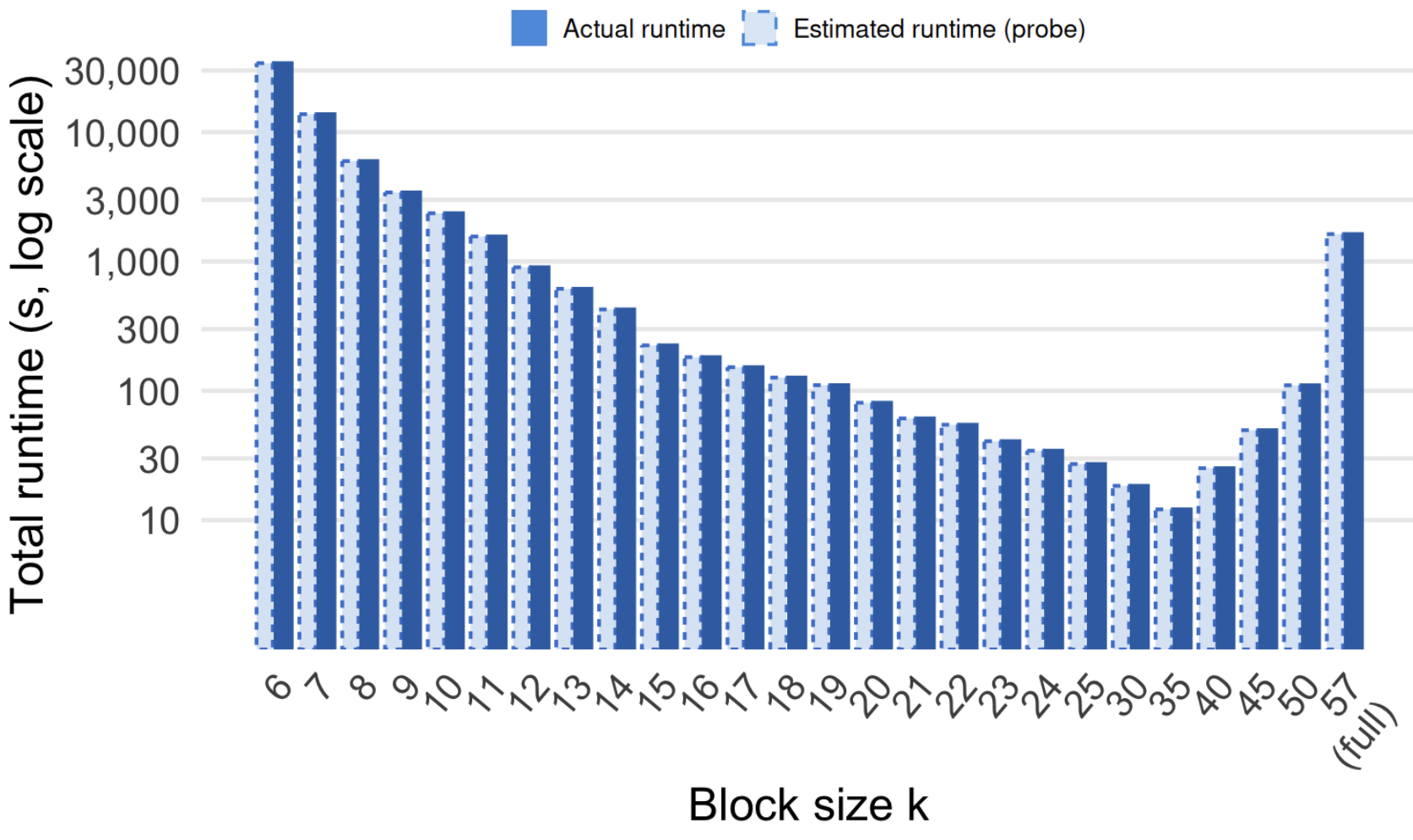}
    \caption{AND gate linear cascade, 7 modules ($|M|=57$, $t=5$).}
    \label{fig:probe-cascade}
  \end{subfigure}

  \caption{Probe-estimated total Hilbert basis runtime as a function of block size~$k$, with support bound fixed at $t=5$. Each point is the product of the per-block trial time (100-iteration probe) and the number of blocks in the covering design. 
  }
  \label{fig:probe-both}
\end{figure}

\section{Benchmarking}
\label{sec:benchmarking}
In this section, we present benchmarks of our new algorithm by applying the covering design strategy on several families of domain-monomer
systems drawn from the thermodynamic binding network and DNA molecular programming
literature~\cite{wang2026molecular,yilmaz2026modularity,sterin2025thermodynamically}.
Our analysis shows that the approach described in \Cref{sec:implementation} can analyze large systems that could not be handled by previous TBN analysis methods or by the naive approach from \cref{thm:HilbertCharacterize}.

We design the experiments to address two key aspects of the covering design approach, namely \emph{scalability} and \emph{equilibrium accuracy}.
Our experiments show that, for many applications,
the covering design strategy with an appropriate value of $t$ and our heuristic choice of $k$ quickly predicts the equilibrium concentration of the system. 
For example, for the AND gate linear cascade system \cite{wang2026molecular} with 7 modules, the covering design approach with $(t,k)=(5,25)$ computes the set $\hat P_{25,5}$ in 23.79s while the full Pareto-optimal set is computed in 1105.05s (\cref{fig:runtime-linear}), achieving a $46\times$ speedup. The concentration analysis with COFFEE shows that the covering design approach does not miss any relevant polymer, as shown in \cref{fig:concentration-error-t5}.

In our experiments, all Hilbert basis computations were performed using \Normaliz 3.11.1.
The Hilbert basis computations dominate the overall time complexity; the other steps, including the heuristic choice of $k$, took relatively little time. Equilibrium concentrations were computed with COFFEE~\cite{COFFEE}. For the AND gate linear cascade experiments, we use two concentration regimes adapted from Figure~S6 of the supplementary materials of Wang et al.~\cite{wang2026molecular}.
In both regimes, the upstream inputs $X_2,\ldots,X_{m+1}$ are supplied at $10\,\mathrm{nM}$.
Regime A sets all non-input monomers to $100\,\mathrm{nM}$.
Regime B uses the same concentrations as Regime A, but additionally supplies $50\,\mathrm{nM}$ of all monomers except the output of the final module.
For each regime, we consider both a correct-output condition, where $X_1$ is supplied at $10\,\mathrm{nM}$, and a leakage condition, where $X_1=0$.\footnote{The theoretical model uses dimensionless mole fractions, but the COFFEE input concentrations are specified in molar units. In the dilute-solution regime considered here, mole fractions and molarities differ by a (temperature-dependent) constant conversion factor.
COFFEE converts between molarity and mole fractions internally.} Each complementary domain bond is assigned energy either $-20\,\mathrm{kcal/mol}$ or $-10\,\mathrm{kcal/mol}$.\footnote{The total energy of a polymer is computed by summing this bond-energy contribution over all complementary domain bonds present in the polymer.}
The value $-20\,\mathrm{kcal/mol}$ represents a strong-binding regime comparable to the saturated regime used by Wang et al.~\cite{wang2026molecular}, while $-10\,\mathrm{kcal/mol}$ gives a weaker-binding system in which unsaturated polymers are more pronounced.\footnote{Although we used uniform bond energies across all domains, our framework does not require this; different domains can have different bond energies.}
All experiments were conducted on a MacBook Pro (Apple M3 Max, 16 cores, 48\,GB RAM).

\subsection{TBN Families}
\label{sec:tbn-families}

To perform numerical experiments and benchmarking, we focused on several families of TBNs.
By ``family,'' we mean a collection of related TBNs in which the same motif is repeated arbitrarily many times.
This parametrization
lets us scale these systems to arbitrarily large sizes, increasing the difficulty of the problem.

 \paragraph*{Linear cascade of AND-gate modules}
The AND-gate module was introduced in prior works \cite{TBNpaper,breik2019computing} and recently realized experimentally as a reversible signal propagation module~\cite{wang2026molecular}.
The linear cascade system~\cite{wang2026molecular} (refer to Figure~S3 of its supplementary materials for a diagram) consists of $m$ sequentially chained AND-gate modules. Each module is a collection of monomers that are normally bonded with each other in a particular way and, in the presence of two input polymers, can rearrange to release an output polymer.
The output of each module serves as one of the two inputs to the next. Ideally, a module's output forms only when both of its inputs are present, implementing an AND-gate computation.
However, unintended reaction pathways cause outputs to form even in the absence of the required inputs, or produce other spurious polymers, a phenomenon known as \emph{leakage}. 
This system is referred to as the AND gate linear cascade system in this paper.

 \paragraph*{Binary tree of AND-gate modules}

 The binary tree system~\cite{wang2026molecular} arranges AND-gate modules in a complete binary tree of depth $d$. Each module operates identically to those in the linear cascade, but here both inputs to every non-leaf module are outputs of two child modules, rather than a mix of primary inputs and upstream outputs. The number of monomers grows exponentially in $d$. This system is referred to as the AND gate binary tree system in this paper.

 \paragraph*{Scaffolded DNA system}
 Sterin et al.~\cite{sterin2025thermodynamically} introduced ``scaffolded DNA computation,'' which obtains the desired output at thermodynamic equilibrium.
 This system differs from the AND module motif in that large polymers are expected: the system consists of a large ``scaffold'' with short monomers bound to it, their number proportional to the length of the scaffold.

 \paragraph*{2-reactant, 2-product reversible signal propagation module (2-2 modules) system}
This module from~\cite{yilmaz2026modularity} behaves similarly to the AND-gate module in the linear cascade, with the difference that each module takes two input monomers and produces two output monomers.
For the benchmarking system, both outputs of a module are used as inputs to the next module in the cascade. It will be referred to as the 2-2 linear cascade in this paper.

\subsection{Runtime Analysis}
\label{sec:runtime}

\cref{fig:runtime-linear,fig:runtime-binary-dna} show the (logarithmic) runtime of our algorithm for the AND gate linear cascade, binary tree system, and 2-2 linear cascade. We vary the system sizes by scaling the depth or the number of modules. In each system, we run the algorithm for various $t$, with the best $k$ chosen by probing (described in \cref{sec:probe-and-prune}).

\begin{figure}[ht]
  \centering
  \includegraphics[width=0.6\textwidth]{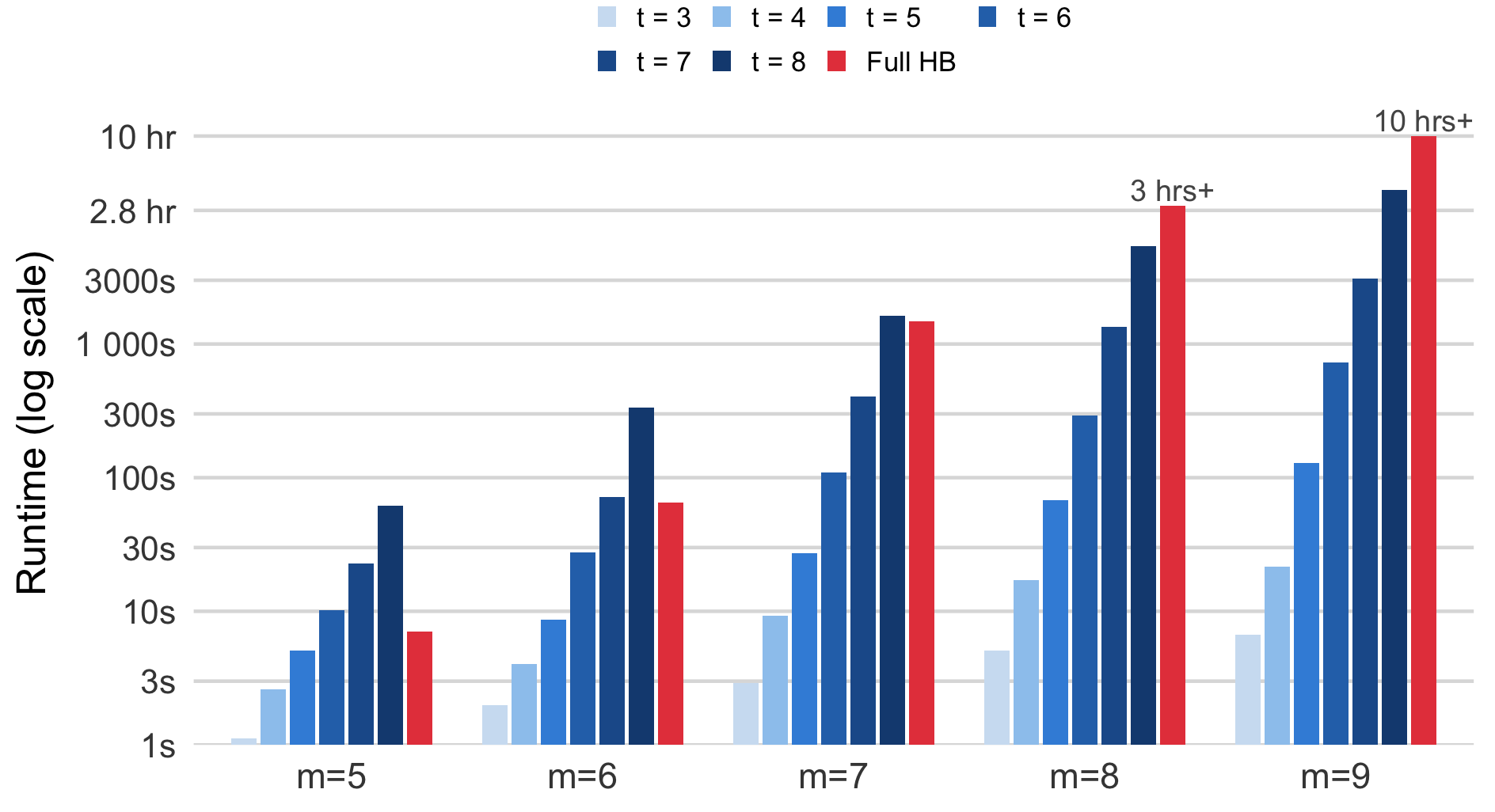}
  \caption{Pareto-optimal polymer set computation time for AND gate linear cascade systems with $m = 5$–$9$ modules. Within each group, bars indicate covering design $t = 3$–$8$ and full Hilbert basis computation (\textsc{Full HB}). Bars labeled \emph{hrs+} denote terminated runs, truncated at the cutoff.
  }
  \label{fig:runtime-linear}
\end{figure}

\begin{figure}[ht]
  \centering
  \includegraphics[width=\textwidth]{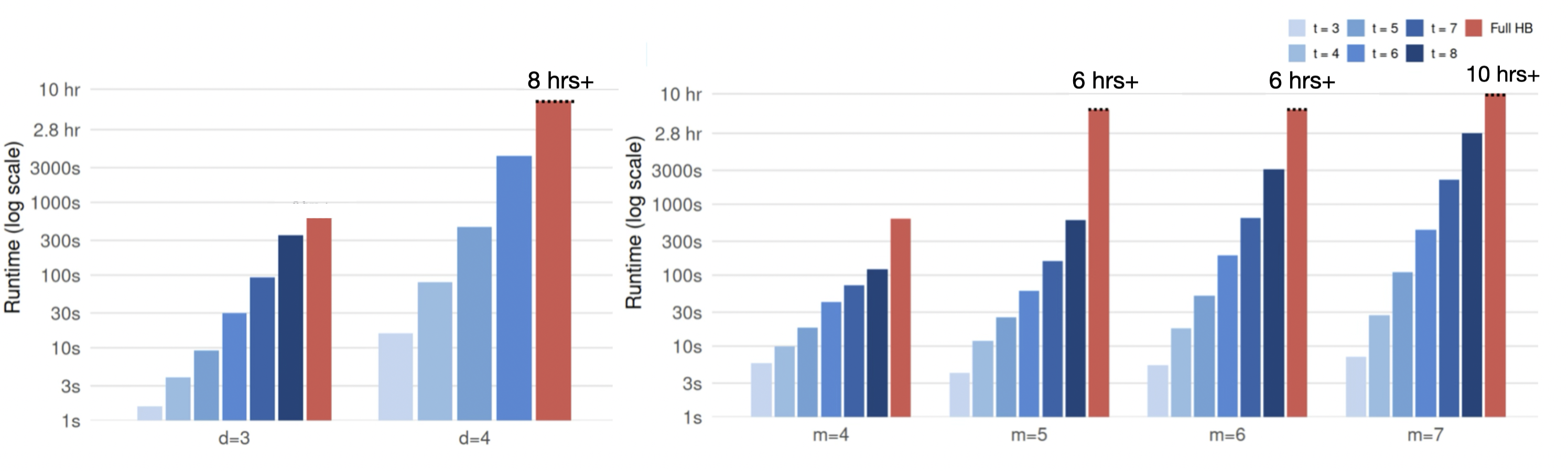}
  \caption{Pareto-optimal polymer set computation time for AND gate binary tree (left) and 2-2 linear cascade (right) TBN systems. Binary tree: depth $d=3,4$, with $t=7,8$ and full HB truncated for $d=4$. 2-2 linear cascade system: $m=4$–$7$, full HB truncated at 6 hours for $m=5,6$ and 10 hours for $m=7$.}
  \label{fig:runtime-binary-dna}
\end{figure}

Across all benchmark families, the covering design strategy provides substantial speedups of several orders of magnitude over full-system Hilbert basis computation (\textsc{Full HB}), which is executed whenever feasible.
The speedup is more drastic for larger systems, while for smaller systems, the covering design strategy with large $t$ shows slower performance.
The Hilbert basis computation dominates the overall runtime for all experiments.

\subsection{Equilibrium Analysis Accuracy}
\label{sec:eqbm-analysis-accuracy}

We compute the equilibrium concentrations using COFFEE~\cite{COFFEE}, given the set of polymers computed using either the covering design approach ($\hat{P}_{k,t}$) or the full Pareto-optimal set ($P^*$). 
For the systems considered in this section, 
the numerical experiments below show that $t=5$ recovers almost all relevant polymers in equilibrium.

We focus on cases where computing the full Pareto-optimal set $P^*$ is feasible to allow comparison with results on $\hat{P}_{k,t}$. 
In particular, we consider the AND gate linear cascade with 7 modules, the AND gate binary tree with depth 3, and the 2-2 linear cascade with 4 modules. The COFFEE analysis with the reduced candidates $\hat P_{k,t}$ for $t=4,\ldots,7$ recovered all significant\footnote{Polymers were categorized as significant if their equilibrium concentration exceeded $0.1\,\mathrm{nM}$ and insignificant otherwise.} polymers from the full set, showing that the support-bounded restriction does not distort the analysis for these systems.

We also consider a seven-module AND gate linear cascade under the Regime B leakage condition with bond energy $E=-20\,\mathrm{kcal/mol}$. Here, the effect of the support bound $t$ becomes visible, as shown in \cref{fig:concentration-error}. For $t=3$, the reduced candidate set misses some equilibrium-relevant polymers and produces significant relative concentration errors for some recovered polymers. In contrast, for $t=5$, all equilibrium-relevant polymers are recovered, and the relative concentration errors are comparatively very low.\footnote{Note that even for $t=3$, only the lower-concentration polymers are missed and the overall equilibrium analysis is still fairly accurate.}

\begin{figure}[ht]
  \centering

  \begin{subfigure}[t]{0.48\textwidth}
    \centering
    \includegraphics[width=\linewidth]{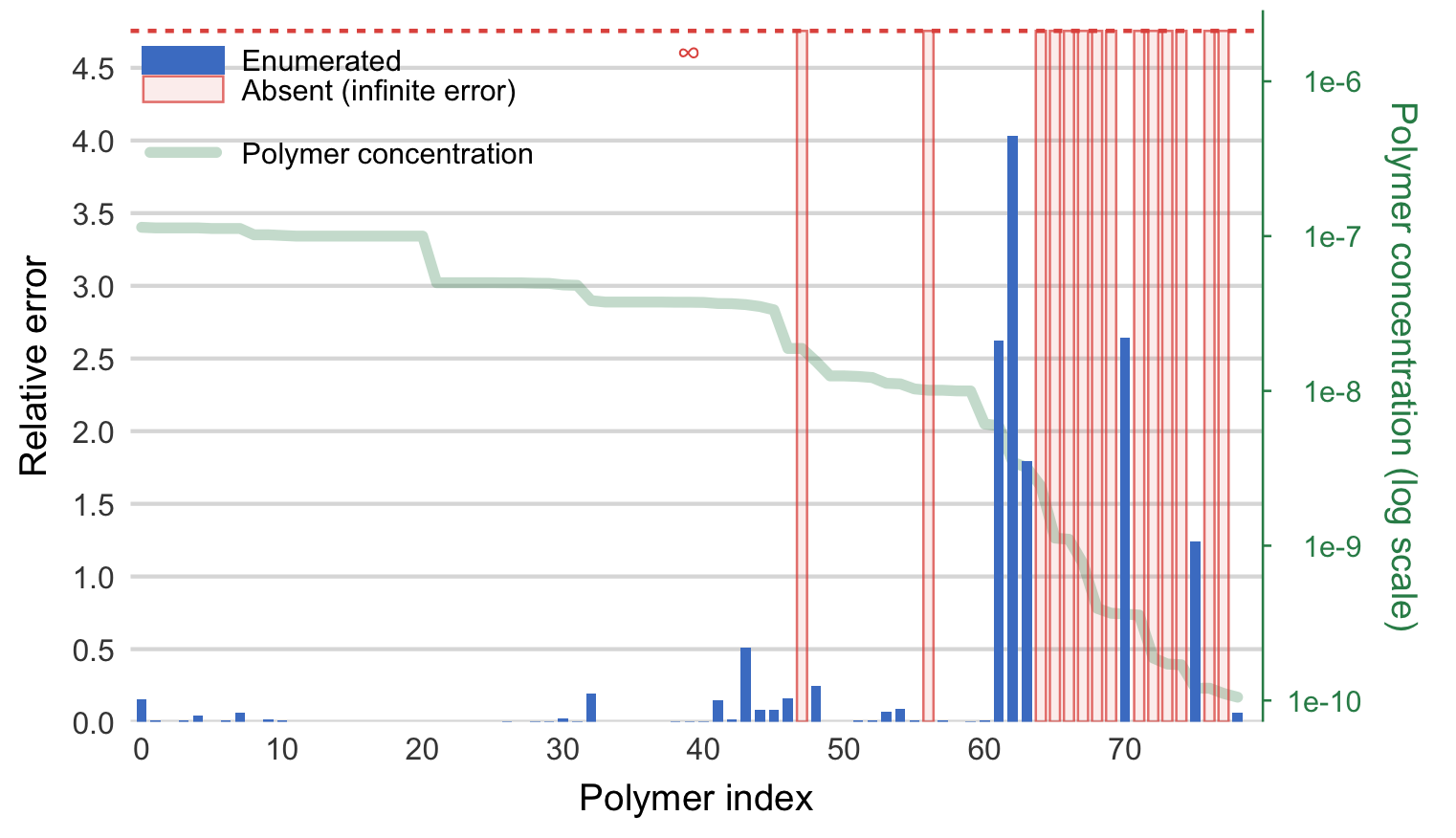}
    \caption{$t=3$ reduced candidate set.}
    \label{fig:concentration-error-t3}
  \end{subfigure}
  \hfill
  \begin{subfigure}[t]{0.48\textwidth}
    \centering
    \includegraphics[width=\linewidth]{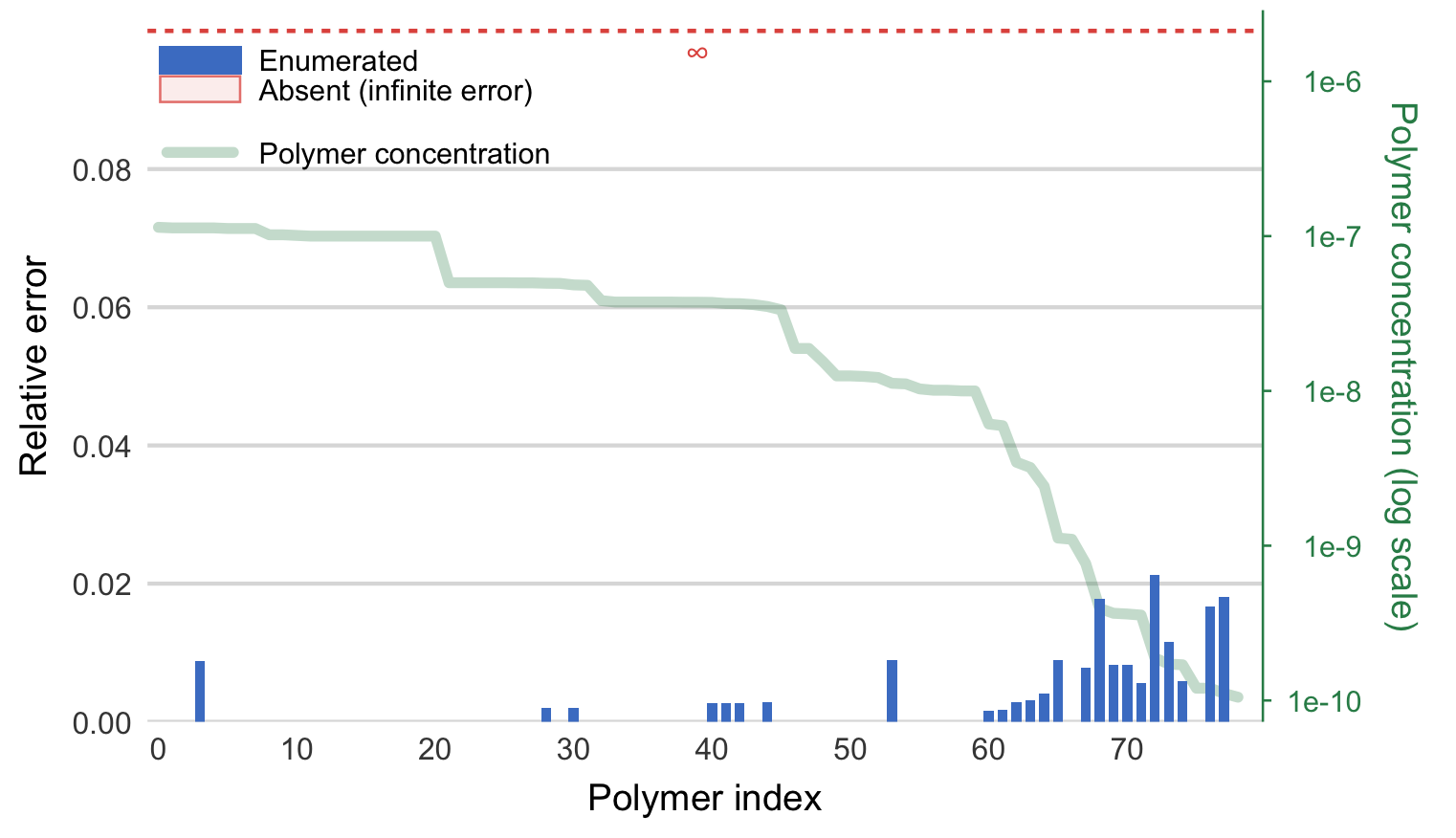}
    \caption{$t=5$ reduced candidate set.}
    \label{fig:concentration-error-t5}
  \end{subfigure}

  \caption{
  Relative concentration error for equilibrium-relevant polymers in the seven-module AND gate linear cascade under the Regime B leakage condition with bond energy $E=-20\,\mathrm{kcal/mol}$.
  Red lines indicate polymers absent from the reduced candidate set, corresponding to infinite relative error; blue lines indicate recovered polymers; the green line shows the polymer concentrations.
  }
  \label{fig:concentration-error}
\end{figure}

\subsection{Leakage Analysis Applications}
\label{sec:leakage-applications}

Beyond equilibrium analysis, the reduced candidate sets can be used for downstream applications as well. One important application is leakage analysis. In engineered TBN systems, leakage refers to the formation of an unintended product through spurious reaction pathways.
We focus on the output leakage for the AND gate linear cascade family here, which is the final output concentration when any one of the inputs is removed.\footnote{The mechanism of an AND gate linear-cascade module is illustrated in Figure~\ref{fig:linear-cascade-module}; a full four-module example is given in Appendix~D.1.}
The output leakage analysis of this system with more modules is out of reach for the previous analytic tools.\footnote{We note that the numerical analysis of~\cite{wang2026molecular} (Section S4.2) restricts to the fully bound (saturated) regime, and even with that simplification only extends up to 10 modules.}

We consider two leakage experiments.
The first experiment compares how much leakage occurs depending on where the input is removed.
We fix a seven-module AND gate linear cascade and remove one input at a time, using Regime B (introduced earlier in \cref{sec:benchmarking}) with bond energy $E=-10\,\mathrm{kcal/mol}$.
For each removed input, the equilibrium concentration of the final output polymer is computed using COFFEE on the candidate set produced by our covering-design enumeration pipeline. \cref{fig:leakage-position} shows the final output concentration variation for this experiment across different values of $E$ and concentration regimes A and B.

\begin{figure}[ht]
    \centering

    \begin{subfigure}[t]{0.48\linewidth}
        \centering
        \includegraphics[width=\linewidth]{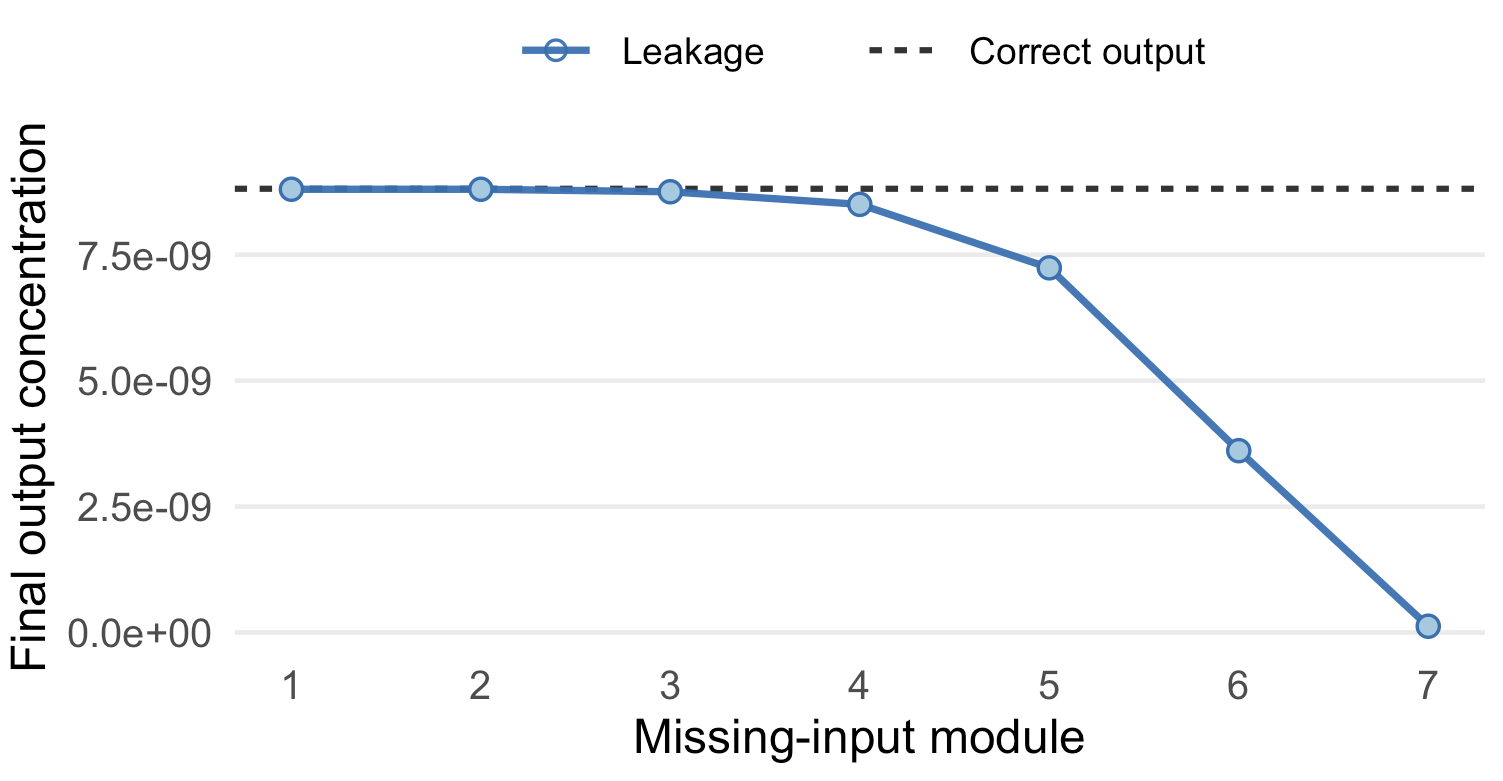}
        \caption{Regime A, $E=-20\,\mathrm{kcal/mol}$}
        \label{fig:leakage-position-regimeA-E20}
    \end{subfigure}
    \hfill
    \begin{subfigure}[t]{0.48\linewidth}
        \centering
        \includegraphics[width=\linewidth]{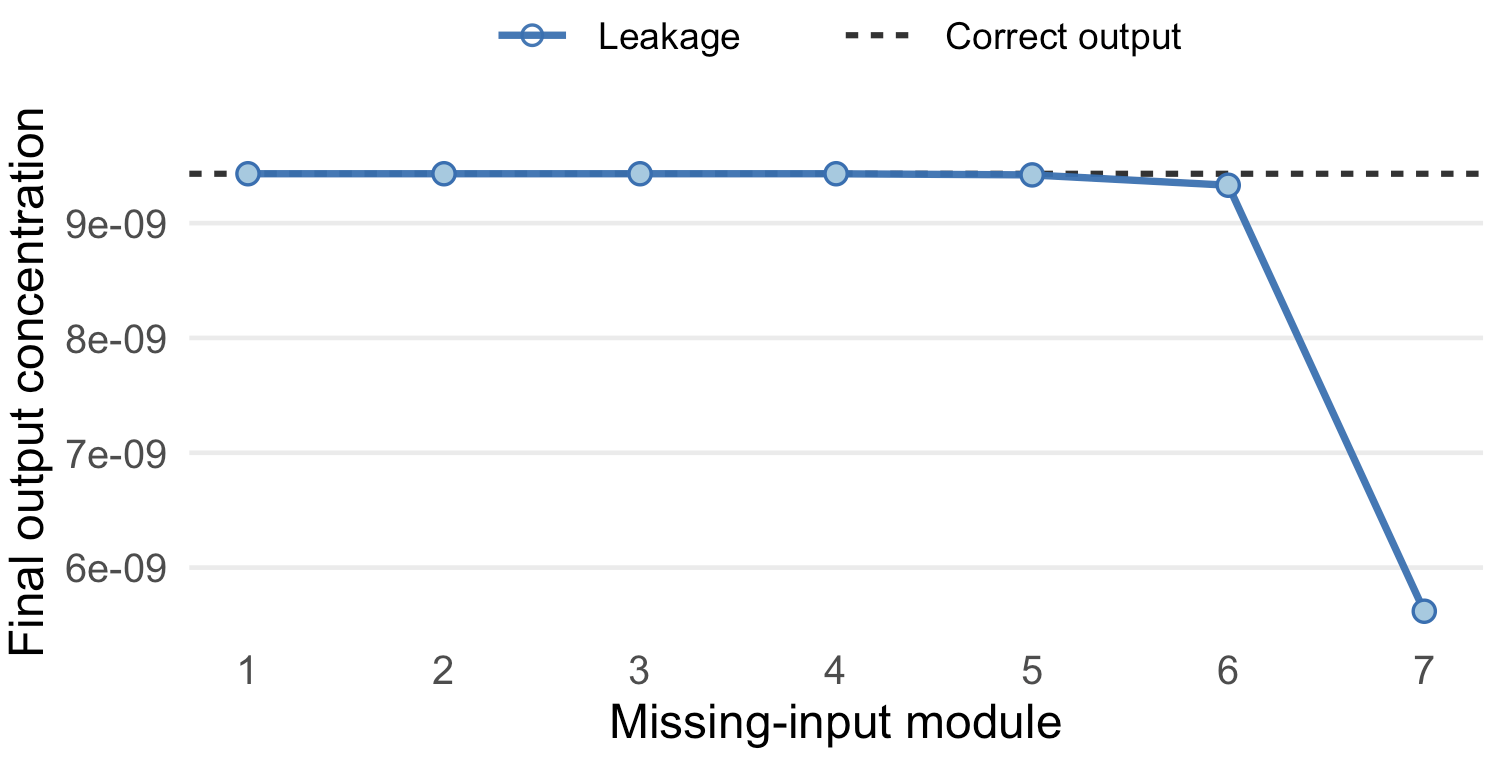}
        \caption{Regime A, $E=-10\,\mathrm{kcal/mol}$}
        \label{fig:leakage-position-regimeA-E10}
    \end{subfigure}

    \vspace{0.6em}

    \begin{subfigure}[t]{0.48\linewidth}
        \centering
        \includegraphics[width=\linewidth]{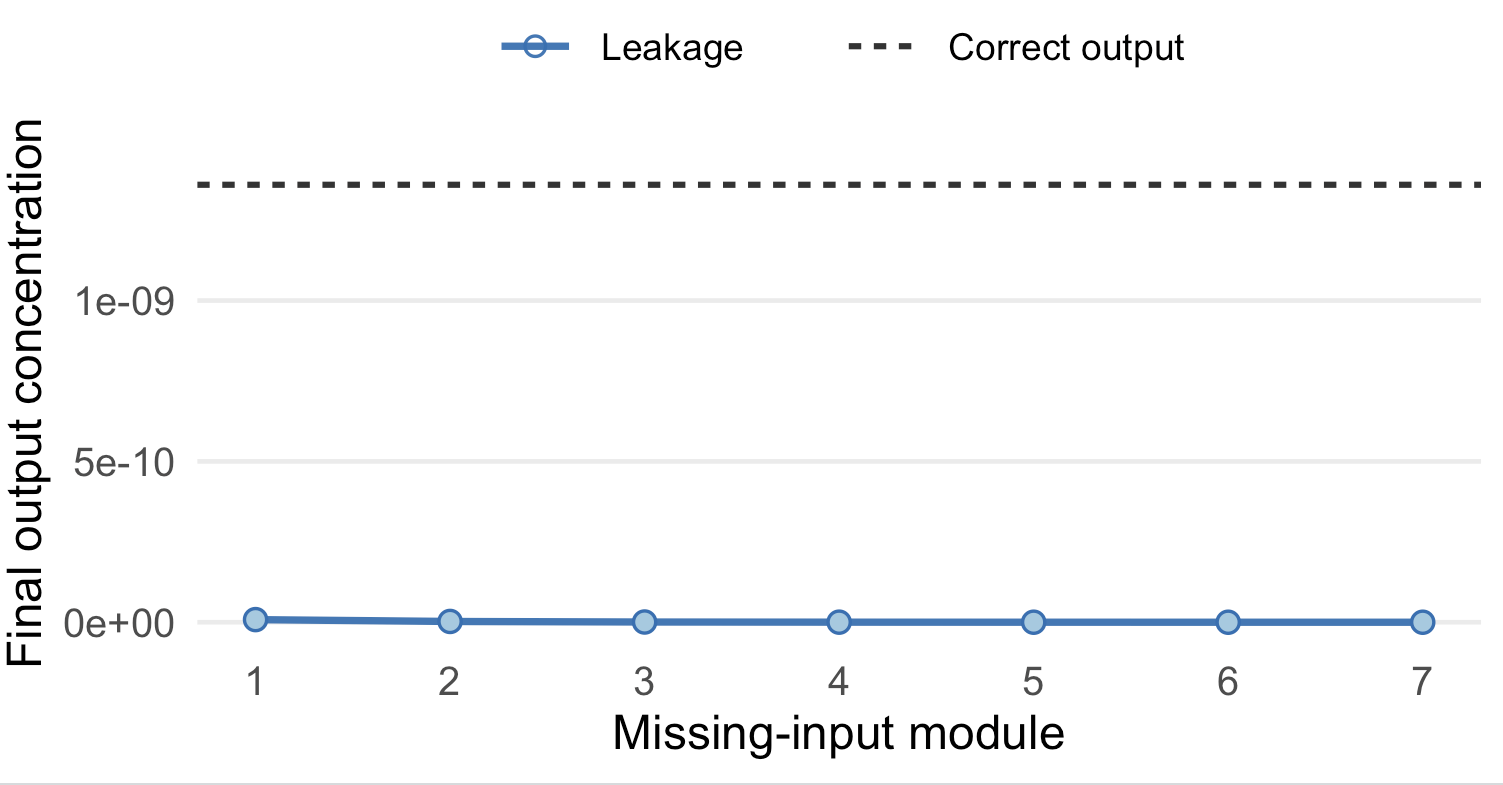}
        \caption{Regime B, $E=-20\,\mathrm{kcal/mol}$}
        \label{fig:leakage-position-regimeB-E20}
    \end{subfigure}
    \hfill
    \begin{subfigure}[t]{0.48\linewidth}
        \centering
        \includegraphics[width=\linewidth]{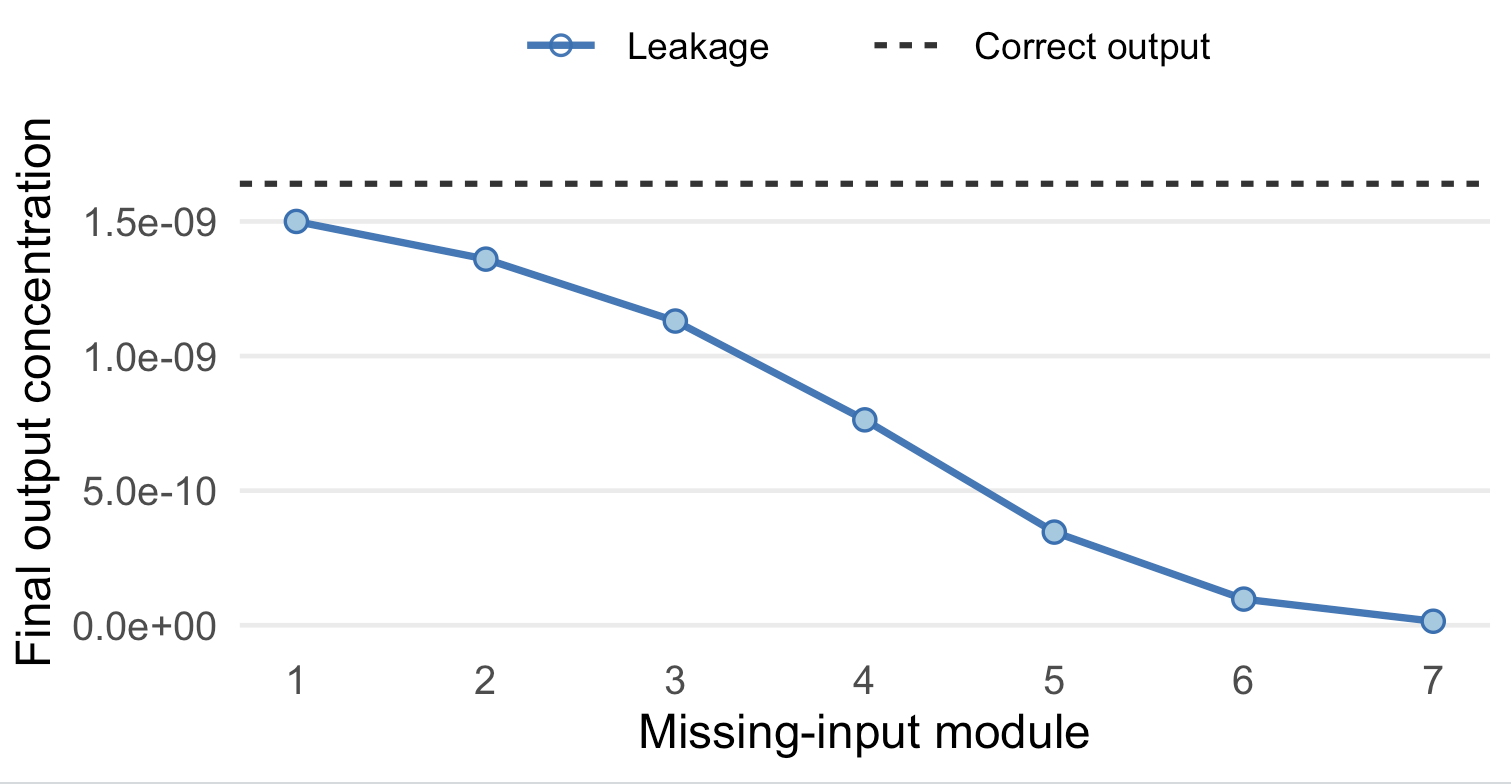}
        \caption{Regime B, $E=-10\,\mathrm{kcal/mol}$}
        \label{fig:leakage-position-regimeB-E10}
    \end{subfigure}

    \caption{
    Final output concentration under single-input deletion in a seven-module linear cascade. 
    The $x$-axis indicates the module whose input is removed. 
    The solid curve shows the leakage concentration, and the dashed horizontal line shows the correct-output concentration when all inputs are present. 
    The four panels compare Regimes A and B under bond energies $E=-20\,,-10\,\mathrm{kcal/mol}$.
    }
    \label{fig:leakage-position}
\end{figure}

As observed in \cref{fig:leakage-position}, removing an input from the first module gives the largest leakage, and the leakage decreases as the removed input moves later in the cascade.
This suggests that leakage generated earlier in the cascade can propagate through subsequent modules, whereas leakage introduced closer to the end has fewer downstream opportunities to amplify.

The second experiment studies how leakage changes as the cascade becomes longer. We remove an input to the first module and vary the number of modules in the linear cascade. For each setting, we compare the leakage output concentration against the correct output concentration obtained when all inputs are present.

\begin{figure}[ht]
    \centering

    \begin{subfigure}[t]{0.48\linewidth}
        \centering
        \includegraphics[width=\linewidth]{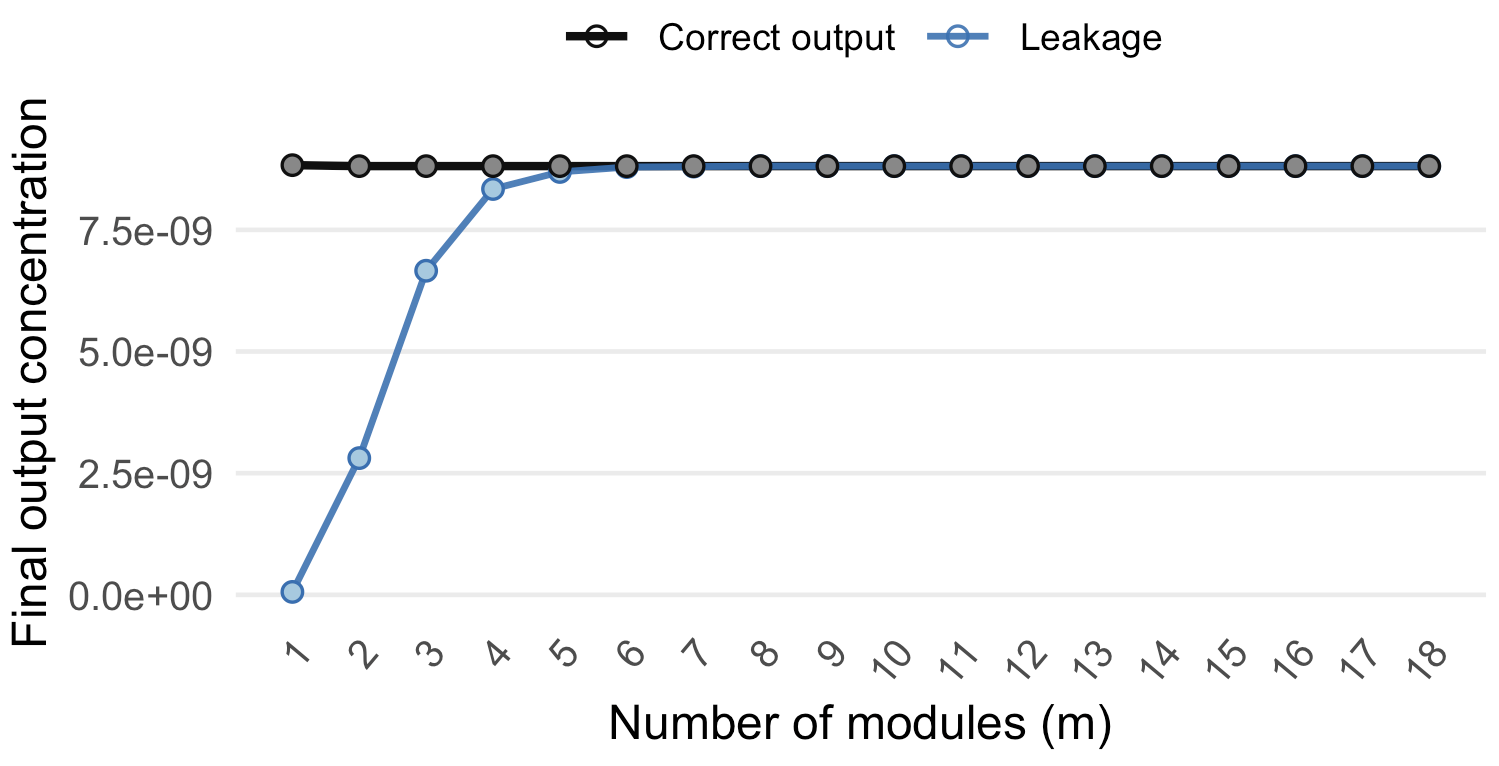}
        \caption{Regime A, $E=-20\,\mathrm{kcal/mol}$}
        \label{fig:leakage-regimeA-E20}
    \end{subfigure}
    \hfill
    \begin{subfigure}[t]{0.48\linewidth}
        \centering
        \includegraphics[width=\linewidth]{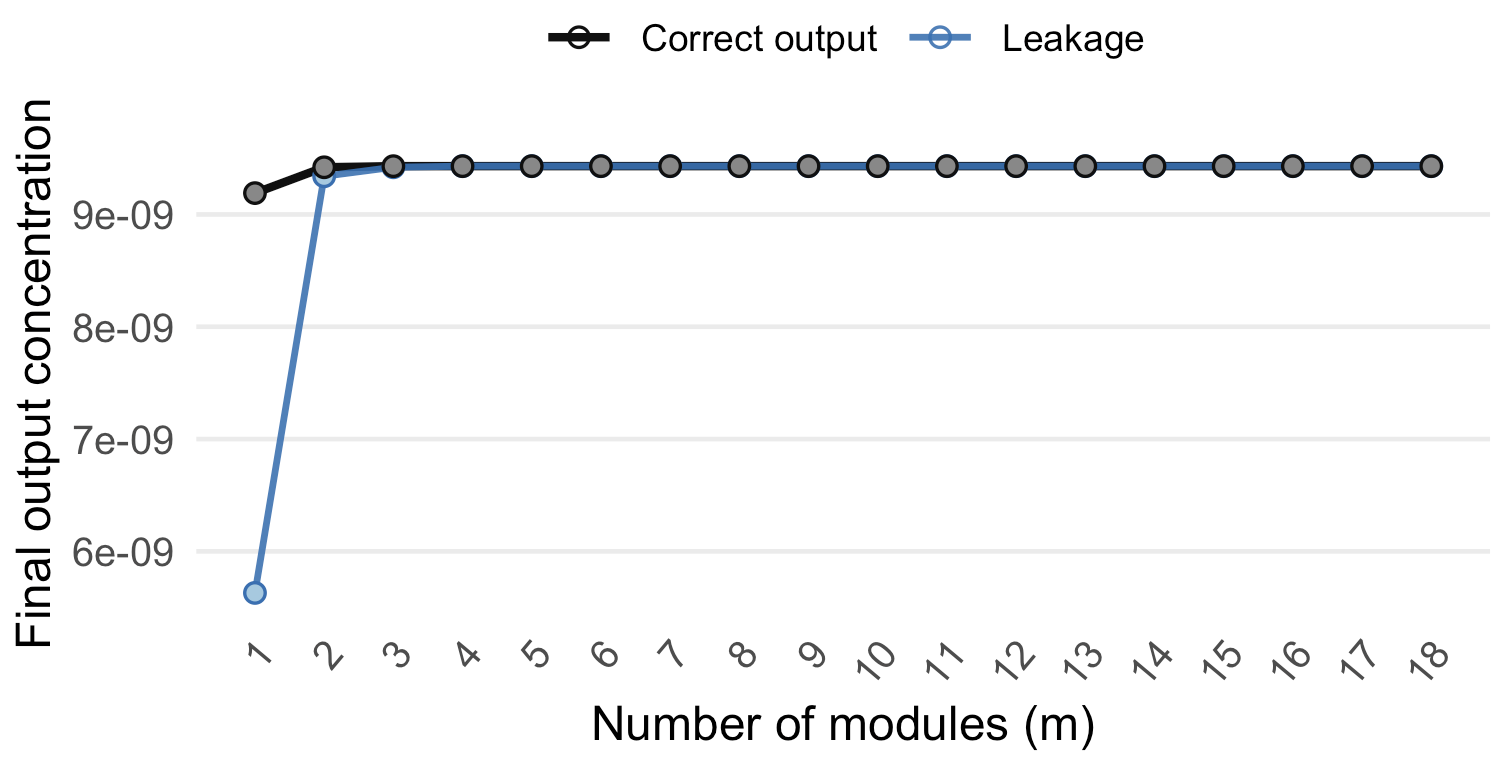}
        \caption{Regime A, $E=-10\,\mathrm{kcal/mol}$}
        \label{fig:leakage-regimeA-E10}
    \end{subfigure}

    \vspace{0.5em}

    \begin{subfigure}[t]{0.48\linewidth}
        \centering
        \includegraphics[width=\linewidth]{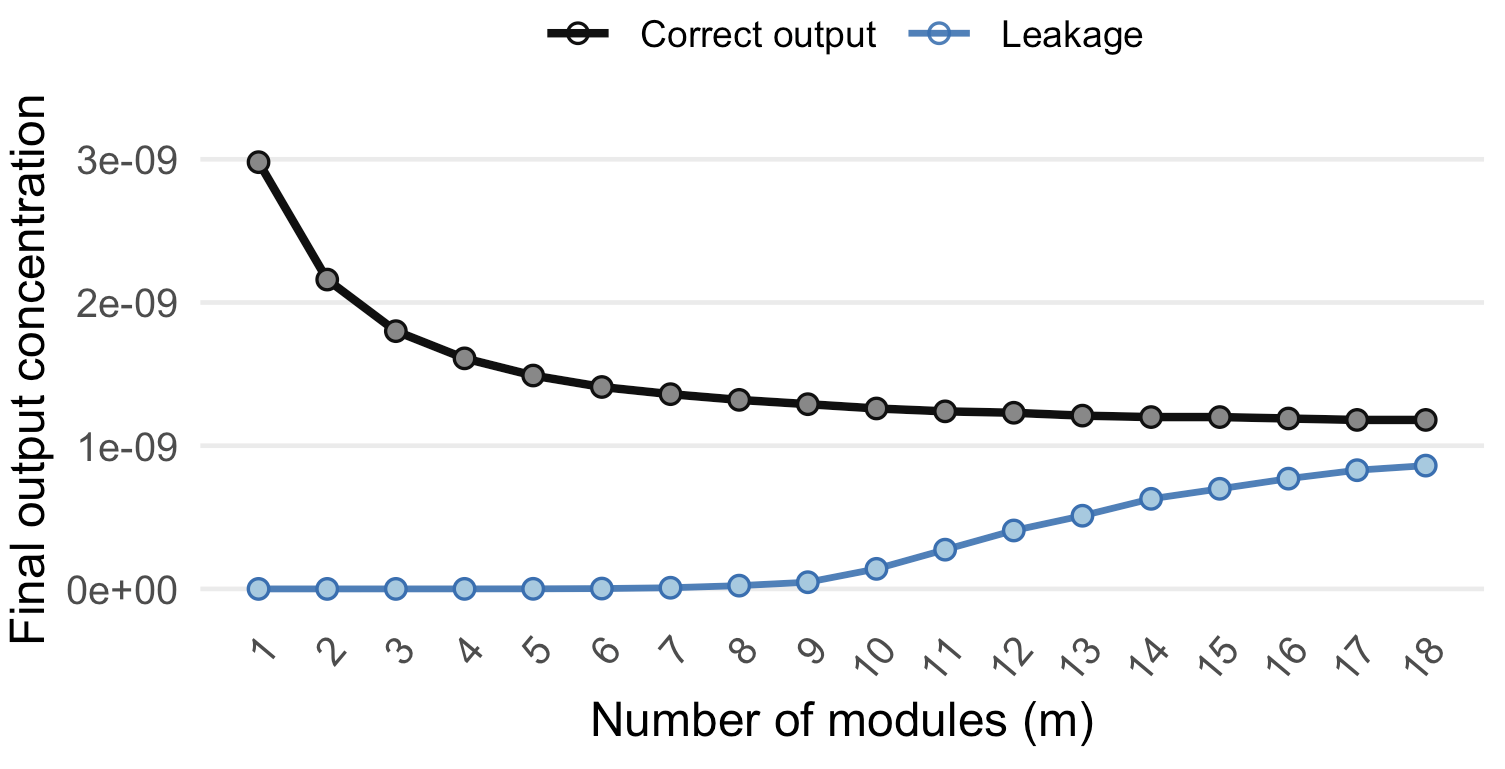}
        \caption{Regime B, $E=-20\,\mathrm{kcal/mol}$}
        \label{fig:leakage-regimeB-E20}
    \end{subfigure}
    \hfill
    \begin{subfigure}[t]{0.48\linewidth}
        \centering
        \includegraphics[width=\linewidth]{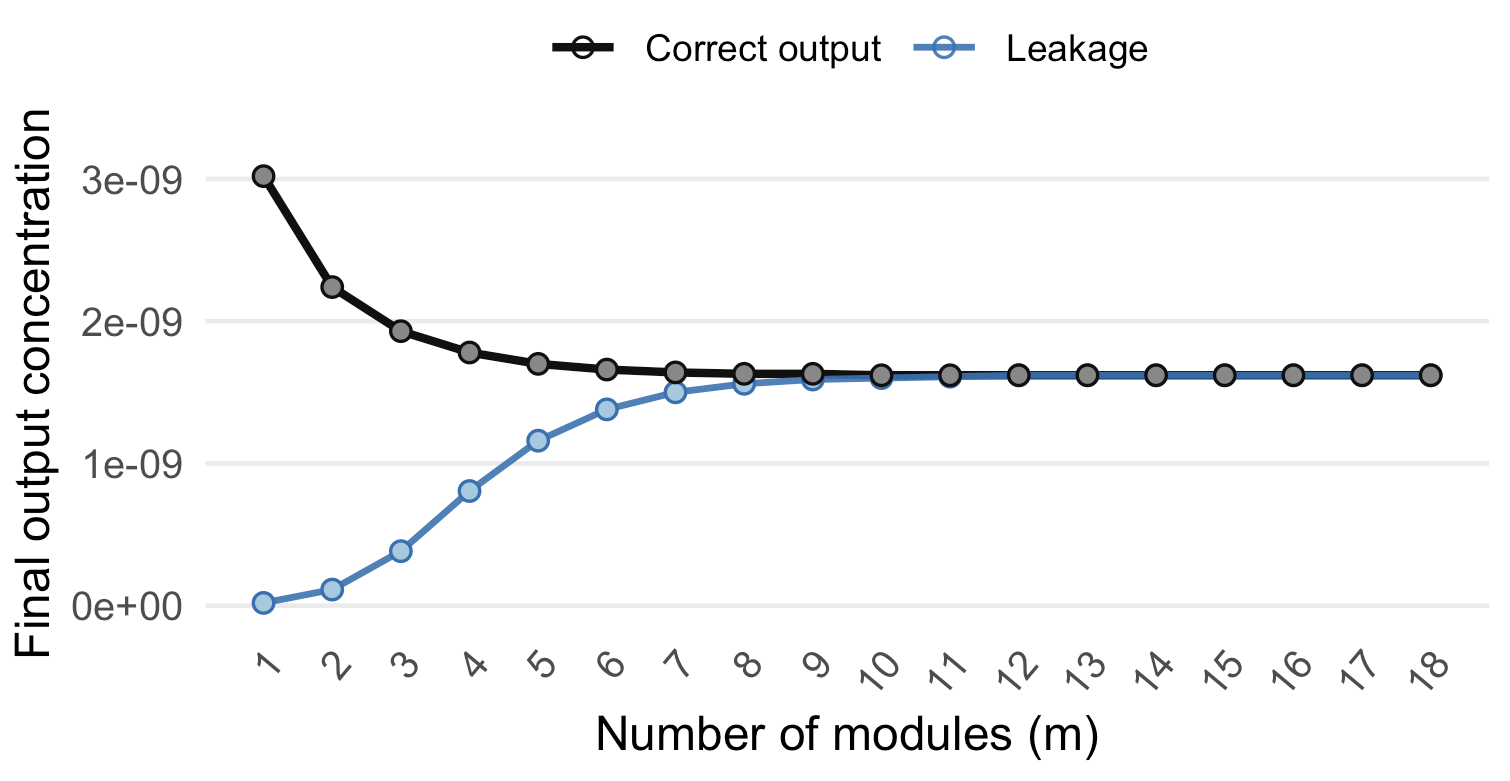}
        \caption{Regime B, $E=-10\,\mathrm{kcal/mol}$}
        \label{fig:leakage-regimeB-E10}
    \end{subfigure}

    \caption{
    Leakage versus correct-output concentration for AND gate linear cascade systems across cascade length. The four panels show the two concentration regimes, A and B, crossed with bond energies $E=-20\,\mathrm{kcal/mol}$ and $E=-10\,\mathrm{kcal/mol}$. For leakage, an input to the first module is removed; for correct output, all inputs are present. All curves are computed with COFFEE using $\hat{P}_{k,t}$, with $t=5$ when feasible and $t=4$ otherwise.
    }
    \label{fig:leakage-analysis}
\end{figure}

We run this experiment under both concentration regimes (A and B) and both bond-energy values (-10 and -20 kcal/mol).
Regime A gives the direct leakage setting, but the final output concentration reaches its plateau quickly.
Regime B suppresses leakage and makes the dependence on cascade length more visible.
Similarly, $E=-20\,\mathrm{kcal/mol}$ mirrors the strong-binding regime used in Wang et al.~\cite{wang2026molecular} (the model used in~\cite{wang2026molecular}, Section S4.2, restricts to saturated polymers only), while $E=-10\,\mathrm{kcal/mol}$ tests a weaker-binding regime where unsaturated polymers play a larger role.
Thus, the four cases in Figure~\ref{fig:leakage-analysis} compare how leakage scales across both concentration and binding regimes.

As shown in Figure~\ref{fig:leakage-analysis}, the final output concentration increases with cascade length.
Regime A reaches a plateau quickly, whereas Regime B has comparatively less leakage.
Weakening the bond energy from $-20\,\mathrm{kcal/mol}$ to $-10\,\mathrm{kcal/mol}$ increases leakage as expected.

Although the runtime still grows roughly exponentially with the system size, the covering-design approach remains practical for systems beyond the range of full Hilbert basis enumeration. When larger systems become intractable, the support parameter $t$ can be decreased to further reduce the enumeration cost. This sacrifices some accuracy in the equilibrium analysis because fewer candidate polymers are included, but the earlier experiments show that small values of $t$ can still recover the dominant equilibrium behavior in the tested systems. The main additional limitation is the size of the covering design: 
the dynamic-programming construction used for large parameters does not guarantee a minimum-size covering design, and may not scale optimally for very large $|M|$. 
Nevertheless, in our experiments, AND gate linear cascade systems with up to eighteen modules were analyzed efficiently enough to support leakage studies that would be infeasible with direct full Hilbert basis computation.

\section{Conclusion}
\label{sec:conclusion}
We introduced a framework for restricting polymer enumeration in geometry-free domain-monomer systems to the set of \emph{Pareto-optimal polymers} for the purpose of equilibrium analysis.
The Hilbert basis analysis provides a principled way to reduce the infinitely many possible polymers to finitely many that dominate the equilibrium. 
Further, because our method permits restricting the number of monomer types in allowed polymers, our algorithm can scale to larger systems for which enumerating all Pareto-optimal polymers would be impractical.
Benchmarking our analysis on AND gate linear cascades, binary trees, and 2-2 linear cascade families from \cite{wang2026molecular,yilmaz2026modularity,sterin2025thermodynamically} shows the efficiency and accuracy of the algorithm for appropriate parameters.

The scalability of our algorithm is based on focusing on polymers with at most $t$ distinct monomer types. This gives users a practical knob for trading completeness against runtime, tied to a design-level quantity rather than to an internal feature of the Hilbert-basis computation. Our experiments show that in many applications, modest values of $t$ suffice to accurately predict the equilibrium behavior.
However, when working with systems in which very large polymers are expected (as with the scaffolded DNA computer~\cite{sterin2025thermodynamically}), this support-bounded analysis is less effective.

Despite the greater efficiency afforded by our algorithm,
our approach inevitably runs into limits for large $t$.
One area for optimization is pre-computation of covering designs for larger system sizes.
Recall that when a precomputed covering design is unavailable~\cite{gordon_2026_19735294}, we use the dynamic-programming construction of Gordon et al.~\cite{gordon1995new}. As described earlier, this construction does not necessarily produce an optimal (smallest possible) covering design.
Nonetheless, even with optimizations, the Hilbert basis algorithm becomes quickly impractical as the system size grows.
The broader open question is whether, beyond the Pareto-optimality criterion introduced here, there are other structural properties---perhaps exploiting additional knowledge of a given system---that could further reduce the candidate set while preserving all equilibrium-relevant polymers.

\bibliography{ref}
\appendix
\section{Deferred Proofs}
\subsection{Proof of \texorpdfstring{\cref{thm:suboptimal-bound}}{Theorem 6}}\label{app:proofbound}
\begin{proof}
Let $\mathbf{p} \in P \setminus P^*$ be a Pareto-suboptimal polymer. By definition, there exists a (nontrivial) split $\mathbf{p} = \mathbf{p}_1 + \mathbf{p}_2$ with $\mathbf{p}_1$ and $\mathbf{p}_2$ not complementary. If either piece is itself suboptimal, split it again, continuing until all remaining polymers are Pareto-optimal. This process terminates because each split strictly decreases the number of monomers in each piece.
Eventually this yields a decomposition of the form:
\[
\mathbf{p} = \boldsymbol{\mu}_1 + \cdots + \boldsymbol{\mu}_k, \qquad k \geq 2, \qquad \boldsymbol{\mu}_i \in P^*,
\]
where every split along the way is non-complementary.
So we have the equilibrium relationship:\footnote{For example, consider the split $\vecp = \vecmu_1+\vecp_2=\vecmu_1+(\vecmu_2+\vecmu_3)$ where $\vecmu_1,\vecmu_2,\vecmu_3$ are Pareto-optimal.
The two intermediate splits $\vecp=\vecmu_1+\vecp_2$ and $\vecp_2=\vecmu_2+\vecmu_3$ are non-complementary, so
we have $[\vecp_2]=[\vecmu_2]\cdot[\vecmu_3]$ and $[\vecp]=[\vecmu_1]\cdot[\vecp_2]$ by \cref{eq:factorization}, which gives $[\vecp]=[\vecmu_1][\vecmu_2][\vecmu_3]$ as desired.}
\begin{equation}\label{eq:product}
[\mathbf{p}] = \prod_{i=1}^{k} [\boldsymbol{\mu}_i].
\end{equation}
Fix one such decomposition for each $\mathbf{p} \in P \setminus P^*$
(although a suboptimal polymer may admit several decompositions).
Note that, after fixing the decompositions, the map
$
\mathbf{p}\mapsto (\boldsymbol{\mu}_1,\dots,\boldsymbol{\mu}_k)
$
is an injective map from the set of Pareto-suboptimal polymers to the set of tuples of Pareto-optimal polymers.
Substituting~\eqref{eq:product} gives
\[
C_{\mathrm{sub}}
= \sum_{\mathbf{p} \in P \setminus P^*} [\mathbf{p}]
= \sum_{\mathbf{p} \in P \setminus P^*} \prod_{i=1}^{k(\mathbf{p})} [\boldsymbol{\mu}_i(\mathbf{p})].
\]
Each suboptimal polymer $\mathbf{p}$ contributes exactly one term to this sum, corresponding to its chosen tuple $(\boldsymbol{\mu}_1, \dots, \boldsymbol{\mu}_k)$ with $k \geq 2$. Since every such tuple is in particular an ordered tuple of elements of $P^*$, and all terms are non-negative, we may bound the sum by ranging over all ordered tuples of $k \geq 2$ elements of $P^*$:\footnote{For example, suppose $P^* = \{\mathbf{a}, \mathbf{b}, \mathbf{c}\}$ and the suboptimal polymers are $\mathbf{a}+\mathbf{b}$, $\mathbf{a}+\mathbf{c}$, and $\mathbf{a}+\mathbf{b}+\mathbf{c}$ for simplicity. Then $C_{\mathrm{sub}} = [\mathbf{a}][\mathbf{b}] + [\mathbf{a}][\mathbf{c}] + [\mathbf{a}][\mathbf{b}][\mathbf{c}]$, while the bound sums $([\mathbf{a}]+[\mathbf{b}]+[\mathbf{c}])^k$ over $k \geq 2$. The $k=2$ term alone already exceeds the first two contributions, and the $k=3$ term covers the third. Each power includes spurious terms (e.g.\ $[\mathbf{a}]^2$ or $[\mathbf{b}]^3$) that do not correspond to any suboptimal polymer, so the bound is loose, but all extra terms are non-negative.}
\[
C_{\mathrm{sub}}
\leq \sum_{k=2}^{\infty}\; \sum_{\boldsymbol{\mu}_1, \dots, \boldsymbol{\mu}_k \in P^*} \prod_{i=1}^{k} [\boldsymbol{\mu}_i].
\]
The inner sum factorizes as a product of independent sums:
\[
\sum_{\boldsymbol{\mu}_1, \dots, \boldsymbol{\mu}_k \in P^*} \prod_{i=1}^{k} [\boldsymbol{\mu}_i]
= \biggl(\sum_{\boldsymbol{\mu} \in P^*} [\boldsymbol{\mu}]\biggr)^{\!k}
= C_{\mathrm{opt}}^k.
\]
Summing the resulting geometric series for $C_{\mathrm{opt}} < 1$ completes the proof:
\[
C_{\mathrm{sub}} \leq \sum_{k=2}^{\infty} C_{\mathrm{opt}}^k = \frac{C_{\mathrm{opt}}^2}{1 - C_{\mathrm{opt}}}.
\qedhere
\]
\end{proof}

\subsection{Proof for Thermodynamic Justification}\label{app:proofTheormo}

\begin{proof}[Proof of \cref{lem:paretosub-equilibrium}]
Let $C$ be a configuration containing a Pareto-suboptimal polymer $\mathbf{x}$. Then there exists a split $\mathbf{x} = \mathbf{x}_1 + \mathbf{x}_2$ with nonzero $\mathbf{x}_1, \mathbf{x}_2 \in \mathbb{N}^M$ such that $\mathbf{x}_1$ and $\mathbf{x}_2$ are not complementary. Let $C'$ be the configuration obtained from $C$ by replacing one copy of $\mathbf{x}$ with one copy each of $\mathbf{x}_1$ and $\mathbf{x}_2$. Since no bond is broken in this split, $H(C') = H(C)$. Since the number of separate polymer units increases by one, $S(C') = S(C) + 1$. Therefore
\[
  G(C') \;=\; a \cdot H(C') - b \cdot S(C') \;=\; G(C) - b \;<\; G(C),
\]
so $C'$ has strictly lower free energy than $C$. Since this holds for any configuration containing $\mathbf{x}$, no such configuration can be a minimum free energy configuration.
\end{proof}

\subsection{Proof of \texorpdfstring{\cref{thm:HilbertCharacterize}}{Theorem 9}}\label{app:proofofHilbertCharacterize}

We observe the following lemma first.
\begin{lemma}\label{lem:lifting_hilbert}
    If $\vecx=\pi(\vech) \in \pi(H)\setminus \{\veczero\}$, then $\vech=\nu(\vecx)$.
\end{lemma}
\begin{proof}
    Choose $\vech\in H$ such that $\vecx = \pi(\vech)$.
    Observe that $\vece_{u_a}+\vece_{u_{a^*}} \in \mathcal{C}(A')$ for all $a\in D$.
    We will show that $\vech$ is dichotomous, so it is the unique dichotomous neutral lift $\nu(\vecx)$.
    Suppose that $(\vech)_{u_a}\cdot (\vech)_{u_{a^*}}\neq 0$ holds for some $a\in D$.
    Then both of them are positive integers because of $\vech \ge \veczero$, which implies that $\vech -\vece_{u_a}-\vece_{u_{a^*}}\in \mathcal{C}(A')$.
    This means that
    \[\vech = (\vece_{u_a}+\vece_{u_{a^*}}) + (\vech -\vece_{u_a}-\vece_{u_{a^*}})\]
    holds, which contradicts the definition of the Hilbert basis. This shows that $\vech$ must be dichotomous. By the uniqueness in \cref{lem:unique_lifting}, $\vech$ must be $\nu(\vecx)$, so $\nu(\vecx)=\vech \in H.$
\end{proof}

\begin{proof}[Proof of \cref{thm:HilbertCharacterize}]
We prove the two inclusions separately.

\noindent{\bf First inclusion $\pi(H)\setminus\{\veczero\} \subseteq P^*$:}
First, choose $\vecx \in \pi(H)\setminus\{\veczero\}$. Toward contradiction, assume that $\vecx \notin P^*$.
This means that $\vecx$ is Pareto-suboptimal, so that there exists a non-complementary split $\vecx=\vecy+\vecz$ with nonzero $\vecy,\vecz$ such that $(A\vecy)_a\cdot (A\vecz)_a \ge0$ for all $a\in D$. In other words, for each $a$, $(A\vecy)_a,(A\vecz)_a$, as well as $(A\vecx)_a$ are all simultaneously nonnegative (i.e., $\ge0$) or simultaneously nonpositive ($\le 0$).

Because all have the same sign, the neutralization $\nu(\vecx),\nu(\vecy),\nu(\vecz)$ all add only one of $u_a$ or $u_{a^*}$ for each $a$. From this, we have $\nu(\vecx) = \nu(\vecy) +\nu(\vecz)$. However, this is a contradiction because \cref{lem:lifting_hilbert} shows that $\nu(\vecx)\in H$, which must not be written as the sum of two nonzero $\nu(\vecy),\nu(\vecz) \in \mathcal C(A')$. This shows that $\pi(H)\setminus\{\veczero\} \subset P^*$.

\noindent{\bf Second inclusion $P^*\subseteq \pi(H)\setminus\{\veczero\} $:}
Next, choose $\vecx \in P^*$, which is nonzero by definition, and assume that $\vecx \notin \pi(H)\setminus\{\veczero\}$. This shows that $\nu(\vecx) \notin H$, otherwise $\vecx = \pi(\nu(\vecx))\in \pi(H)\setminus \{\veczero\}$.

By the definition of the Hilbert basis, there are two nonzero $\vecv,\vecw \in \mathcal C(A')$ such that $\nu(\vecx)=\vecv+\vecw$, which gives a split $\vecx = \pi(\vecv)+\pi(\vecw)$. We will show that this is a non-complementary split so that $\vecx$ is Pareto-suboptimal.

Given $\vecv,\vecw \ge \veczero$ and the dichotomous property of $\nu(\vecx)$, the split polymers $\vecv,\vecw$ must also be dichotomous.
Furthermore, for each $a\in D$, one of the two unit-monomer coordinates is zero in all three vectors:
\begin{equation}\label{eqn:zero_direct}
(\nu(\vecx))_{u_a}=0 \text{ and } (\vecv)_{u_a}=(\vecw)_{u_a}=0, \text{ or } (\nu(\vecx))_{u_{a^*}}=0 \text{ and } (\vecv)_{u_{a^*}}=(\vecw)_{u_{a^*}}=0.
\end{equation}
Suppose that $\pi(\vecv)=\veczero$ so that $\vecx=\pi(\vecw)$.
The polymer $\vecw$ is a dichotomous neutral element in $\pi^{-1}(\vecx)=\pi^{-1}(\pi(\vecw))$, which must be $\nu(\vecx)$ by \cref{lem:unique_lifting}, leading to a contradiction: $\vecv=\nu(\vecx)-\vecw=\veczero$. Therefore $\pi(\vecv)\neq \veczero$ and similarly $\pi(\vecw)\neq \veczero$.

It remains to show that $\pi(\vecv)$ and $\pi(\vecw)$ are non-complementary. By \cref{eqn:A'expanding},
\[
(A\pi(\vecw))_a + (\vecw)_{u_a}-(\vecw)_{u_{a^*}} = (A'\vecw)_a =0
\]
holds for all $a\in D$, where the last equality is due to $\vecw \in \mathcal C(A').$ By \cref{eqn:zero_direct}, one of the following must be true:
\[
(A\pi(\vecw))_a + (\vecw)_{u_a}=(A\pi(\vecv))_a + (\vecv)_{u_a}=0 \text{ or }
(A\pi(\vecw))_a - (\vecw)_{u_{a^*}}=(A\pi(\vecv))_a - (\vecv)_{u_{a^*}}=0 .
\]
In any case, $(A\pi(\vecw))_a$ and $(A\pi(\vecv))_a$ are both nonnegative or nonpositive because of $\vecv,\vecw\ge \vec0$, so that 
$(A\pi(\vecw))_a\cdot (A\pi(\vecv))_a \ge 0$ for all $a\in D$. This proves that $\pi(\vecv)$ and $\pi(\vecw)$ are non-complementary. Since $\pi(\vecv),\pi(\vecw)\neq \veczero$ and $\vecx=\pi(\vecv)+\pi(\vecw)$, this is a non-complementary split of $\vecx$, contradicting $\vecx\in P^*$. Therefore $\vecx\in \pi(H)\setminus\{\veczero\}$.

Combining the two inclusions, we conclude the proof.
\end{proof}

\subsection{Proof for Covering Design Strategy}\label{app:proofcovring}
\begin{proof}[Proof of \cref{lem:Paretocharacter_embedding}]
    First, if $\vecp \in P^*_S \subset P^*$, then any split $\vecp=\vecp_1+\vecp_2$ for nonzero $\vecp_1,\vecp_2 \in \Supp(S)\subset \N^M$ must be complementary. Therefore, $\vecp$ is a Pareto-optimal polymer in $(\Sigma,S)$ after embedding.

    On the other hand, if $\vecp$ is a Pareto-optimal polymer in $(\Sigma,S)$ after embedding, it is clear that $\vecp \in \Supp(S).$
    Also, any split of it $\vecp=\vecp_1+\vecp_2$ for nonzero $\vecp_1,\vecp_2 \in \Supp(S)$ must be complementary. We showed that all splits must be in this form, thus $\vecp$ is in $P^*$. This proves the equivalence.
\end{proof}

\begin{proof}[Proof of \cref{lem:union_naive}]
    For the first statement, observe that the entries in $M\setminus S$ are always 0 and can be removed. After removal, this is exactly \cref{thm:HilbertCharacterize} on the system $(\Sigma,S)$.

    For the second statement, any $\vecp \in P_t^*$ has $|\supp(\vecp)|\le t$, so $\supp(\vecp)$ is contained in some subset $S\subseteq M$ with $|S|=t$. Hence $\vecp\in P_S^*=\pi(H_S)\setminus\{\veczero\}$. Conversely, if $\vecp\in \pi(H_S)\setminus\{\veczero\}$ for some $S\subseteq M$ with $|S|=t$, then $\vecp\in P_S^*\subseteq P^*$ and $|\supp(\vecp)|\le |S|=t$, so $\vecp\in P_t^*$.
\end{proof}

\begin{proof}[Proof of \cref{thm:covering}]
    The inclusion $\subseteq P^*$ holds because $\pi(H_B)\setminus \{\veczero\}=P_B^* =P^* \cap \Supp(B) \subset P^*.$

    For the first inclusion, it suffices to show that $P_S^* \subseteq  \hat{P}_{k,t} $ for all $|S|=t$ because of \cref{lem:union_naive}.
    By the definition of the covering design and \cref{eqn:inclusion}, there exists $B\in \mathcal D$ such that $S\subseteq B$ so that
    $P_S^* \subseteq P_B^*$. This concludes the proof.
\end{proof}

\section{Scalability Limitations}\label{app:failed}
\cref{thm:HilbertCharacterize} shows that the set of Pareto-optimal polymers $P^*$ is finite and computable.
We can employ a modern implementation, such as \Normaliz, to compute the Hilbert basis $H$, which gives the set of Pareto-optimal polymers for a relatively small domain-monomer system.
However, this approach quickly becomes impractical as the system grows.
We give some explanation of this phenomenon in the next subsection.

A natural way to remedy this inefficiency is to terminate the algorithm early, expecting to obtain a part of the full Hilbert basis. Unfortunately, the state-of-the-art implementations of Hilbert basis algorithms usually enumerate a superset of the Hilbert basis first and then iteratively reduce it to the Hilbert basis.
This algorithmic aspect makes it difficult to introduce a meaningful stopping criterion, as we are not aware of any other way to understand the intermediate superset except for the reduction to the Hilbert basis.

We have made several attempts to employ scalable algorithms, which turned out to be inefficient.
We explain two notable approaches below.

\begin{remark}[Alternative representation]
We can use the cone $\{\mathbf{x} \in \mathbb{N}^{|M|+|D|} \mid A'\mathbf{x} \geq \mathbf{0}\}$ and its Hilbert basis by only adding the complementary unit monomers $\{u_{a^*}\}_{a \in D}$.
In this setting, the net domain vector has nonnegative entries instead of being a zero vector. This formulation and computation ultimately gives the same result as in this paper, but modern implementations such as \Normaliz are more native to the cone $A'\mathbf{x} = \mathbf{0}$, as we present in the main body.
\end{remark}

\subsection{Hilbert Basis Hardness}
The size of the Hilbert basis rapidly increases along several parameters, and counting the number of Hilbert basis elements is known to be NP-hard (strictly speaking, \#P-hard \cite{hermann1999complexity}). In fact, the Hilbert basis of a pointed cone defined by a tiny matrix can be arbitrarily large:
the Hilbert basis of the pointed cone $\mathcal C_{\mathbb R}(A_k)=\{\vecx \in \mathbb R^3 | A_k \vecx =0 , \vecx \ge 0\}$
for $A_k = \begin{bmatrix}
    1&1&-k
\end{bmatrix}$ is
\[
H_k = \{(i,k-i,1)\mid i=0,\ldots,k\}.
\]
In this example, the size of $H_k$ increases exponentially for the bit-length of describing $A_k$, i.e., the input size.

\subsection{Failed Approaches}
This subsection introduces two slow approaches, highlighting that it is not clear how to obtain a practical, scalable algorithm by either (i) employing an older scalable Hilbert basis algorithm or (ii) directly exploiting properties of the Pareto-optimal polymers.
This motivates the main approach in \cref{sec:covering}, employing both modern Hilbert basis algorithms and properties of the Pareto-optimal polymers.

The first approach is based on an old Hilbert basis algorithm~\cite{contejean1994efficient}. This algorithm inspects candidate basis elements by increasing the target $1$-norm (i.e., the sum of absolute values of entries). This is exceedingly slow, especially because of the enormous number of small-norm polymers.
The algorithm reduces the search size using some geometric ideas, but our implementation failed to terminate even for small examples.

Another approach is to use the following nontrivial fact for Pareto-optimal polymers.

\begin{lemma}
For every Pareto-optimal polymer $\vecp \in P^*$ that is not a monomer, there exist nonzero $\vecp_1, \vecp_2 \in P^*$ such that $\vecp = \vecp_1 + \vecp_2$.
\end{lemma}
We will not prove this lemma here because it is not relevant to the main body of this paper.
Given this lemma, we can recursively enumerate $P^*$ by repeatedly merging pairs $\vecp, \vecq$ of previously enumerated Pareto-optimal polymers and checking whether the resulting $\vecp+\vecq$ is Pareto-optimal.
Our implementation (with some optimizations) shows that this approach is also much slower than the simple Hilbert basis algorithm using \cref{thm:HilbertCharacterize}.

\section{Domain Support Mode}
\label{sec:domainmode}

This section introduces an alternative way to make the computation of the Pareto-optimal polymers scalable, based on the domain supports instead of the monomer supports. In the \emph{domain mode}, subset enumeration is applied over domain pairs rather than monomer types.

For a monomer $m$, we define the domain-support of $m$, denoted by $\supp_\Sigma(m)$, as the set of domain-pair coordinates $d\in D$ such that the encoded vector $\veca_m \in \Z^D$ has a nonzero entry at the $d$-th position, i.e., $\veca_m(d)\neq 0$. 

In the \emph{domain mode}, we compute the Pareto-optimal polymers whose monomers are, for some $t$-element subset $T\subseteq D$, all in a set
\[
M_T=\{m\in M\mid \supp_\Sigma(m) \subseteq T\},
\]
i.e., those monomers whose domains all lie within $T$.
Similar to the monomer support case, we compute the set
\[
\hat{P}^{\mathrm{dom}}_{k,t}
=
\bigcup_{B \in \mathcal D}
\left(\pi(H^{\mathrm{dom}}_B) \setminus \{\veczero\}\right),
\]
where $\mathcal D$ is the $(|D|,k,t)$-covering design and $H^{\mathrm{dom}}_B$ is the Hilbert basis for the augmented domain-monomer matrix $A_B' \in \Z^{|B|\times (|M_B|+2|B|)}$ obtained by removing the irrelevant coordinates.

We omit the detailed analysis as it is analogous to the monomer mode.

Domain mode may be advantageous when the number of domains is relatively smaller than the number of monomers, or when many monomers share the same domain pairs, as in cascade and binary-tree TBNs, in which case the relevant covering design size $c(|D|, k, t)$ may be significantly smaller than $c(|M|, k, t)$.

\section{TBN System Examples}

\begingroup
\setlength{\columnsep}{1.2em}
\begin{multicols}{2}

\subsection{Linear Cascade 4 Module AND Gate System}

\begin{center}
  \includegraphics[width=\linewidth]{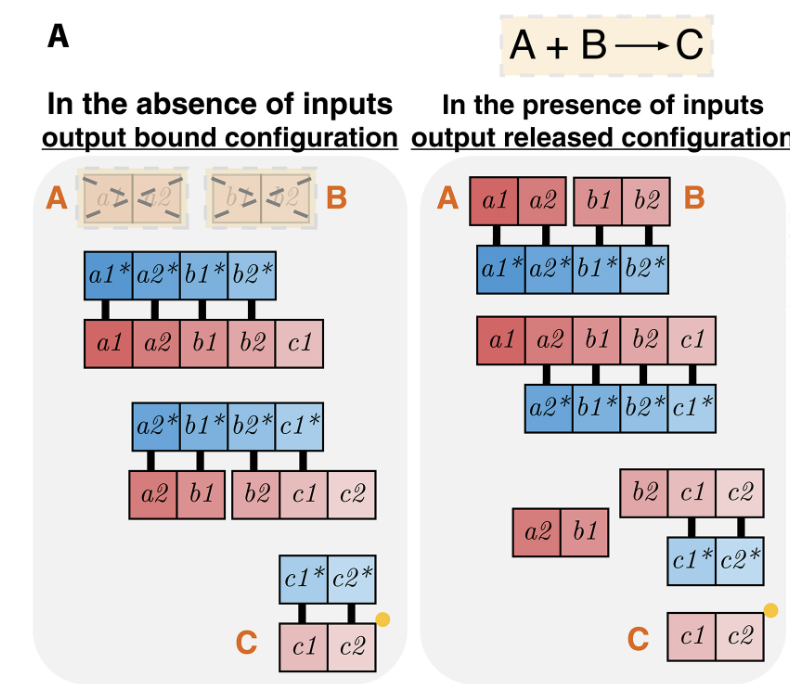}
  \captionsetup{hypcap=false}
  \captionof{figure}{Mechanism of a single AND-gate module in the linear cascade system. In the presence of both input polymers, the module can rearrange to release the output polymer, which then serves as an input to the next module in the cascade. Figure adapted from Fig.~3 of \cite{wang2026molecular}.}
  \label{fig:linear-cascade-module}
\end{center}

\begin{Verbatim}[fontsize=\scriptsize,breaklines=true,breakanywhere=true]
inputs: x1_1 x1_2
        x2_1 x2_2
        x3_1 x3_2
        x4_1 x4_2
        x5_1 x5_2
        
# Module 1: x1 + x2 -> y1
x1_1 x1_2 x2_1 x2_2 y1_1
x1_1* x1_2* x2_1* x2_2*
x1_2 x2_1
x2_2 y1_1 y1_2
x1_2* x2_1* x2_2* y1_1*
y1_1* y1_2*
outputy1: y1_1 y1_2

# Module 2: y1 + x3 -> y2
y1_1 y1_2 x3_1 x3_2 y2_1
y1_1* y1_2* x3_1* x3_2*
y1_2 x3_1
x3_2 y2_1 y2_2
y1_2* x3_1* x3_2* y2_1*
y2_1* y2_2*
outputy2: y2_1 y2_2

# Module 3: y2 + x4 -> y3
y2_1 y2_2 x4_1 x4_2 y3_1
y2_1* y2_2* x4_1* x4_2*
y2_2 x4_1
x4_2 y3_1 y3_2
y2_2* x4_1* x4_2* y3_1*
y3_1* y3_2*
outputy3: y3_1 y3_2

# Module 4: y3 + x5 -> y4
y3_1 y3_2 x5_1 x5_2 y4_1
y3_1* y3_2* x5_1* x5_2*
y3_2 x5_1
x5_2 y4_1 y4_2
y3_2* x5_1* x5_2* y4_1*
y4_1* y4_2*
outputy4: y4_1 y4_2
\end{Verbatim}

\subsection{2-2 Linear Cascade System}
\begin{Verbatim}[fontsize=\scriptsize,breaklines=true,breakanywhere=true]
x11 x12 x13
x21 x22 x23
x11 x12 x13 x21 x22 x23 y11 z11
x11* x12* x13* x21* x22* x23*
x12 x13 x22 x23 y11 y12 z11 z12
x12* x13* x22* x23* y11* z11*
x13 x23 y11 y12 y13 z11 z12 z13
x13* x23* y11* y12* z11* z12*
y11* y12* y13* z11* z12* z13*
y11 y12 y13
z11 z12 z13
y11 y12 y13 z11 z12 z13 y21 z21
y11* y12* y13* z11* z12* z13*
y12 y13 z12 z13 y21 y22 z21 z22
y12* z13* z12* z13* y21* z21*
y13 z13 y21 y22 y23 z21 z22 z23
y13* z13* y21* y22* z21* z22*
y21* y22* y23* z21* z22* z23*
y21 y22 y23
z21 z22 z23
\end{Verbatim}

\subsection{Scaffolded DNA System}
\begin{Verbatim}[fontsize=\scriptsize,breakanywhere=true]
s[1]* s[2]* s[3]* s[4]* s[5]* s[6]* s[7]* s[8]* s[9]* s[10]*
t0[1] s[1] t0[2]*
t1[1] s[1] t1[2]*
t0[2] s[2] t0[3]*
t1[2] s[2] t1[3]*
t0[3] s[3] t0[4]*
t1[3] s[3] t1[4]*
t0[4] s[4] t0[5]*
t1[4] s[4] t1[5]*
t0[5] s[5] t0[6]*
t1[5] s[5] t1[6]*
t0[6] s[6] t0[7]*
t1[6] s[6] t1[7]*
t0[7] s[7] t0[8]*
t1[7] s[7] t1[8]*
t0[8] s[8] t0[9]*
t1[8] s[8] t1[9]*
t0[9] s[9] t0[10]*
t1[9] s[9] t1[10]*
t0[10] s[10]
t1[10] s[10]
\end{Verbatim}

\subsection{AND Gate Binary Tree System}
\begin{Verbatim}[fontsize=\scriptsize,breaklines=true,breakanywhere=true]
input: x1_1 x1_2, x2_1 x2_2, x3_1 x3_2, x4_1 x4_2
# Module 1: x1 + x2 -> y1
x1_1 x1_2 x2_1 x2_2 y1_1
x1_1* x1_2* x2_1* x2_2*
x1_2 x2_1
x2_2 y1_1 y1_2
x1_2* x2_1* x2_2* y1_1*
y1_1* y1_2*
outputy1: y1_1 y1_2
# Module 2: x3 + x4 -> y2
x3_1 x3_2 x4_1 x4_2 y2_1
x3_1* x3_2* x4_1* x4_2*
x3_2 x4_1
x4_2 y2_1 y2_2
x3_2* x4_1* x4_2* y2_1*
y2_1* y2_2*
outputy2: y2_1 y2_2
# Module 3: y1 + y2 -> y3
y1_1 y1_2 y2_1 y2_2 y3_1
y1_1* y1_2* y2_1* y2_2*
y1_2 y2_1
y2_2 y3_1 y3_2
y1_2* y2_1* y2_2* y3_1*
y3_1* y3_2*
outputy3: y3_1 y3_2
\end{Verbatim}

\end{multicols}
\endgroup
\end{document}